\documentclass[11pt]{amsart}
\usepackage[letterpaper,total={153mm,215mm},hmarginratio=1:1,footskip=12mm]{geometry} 
\makeatletter 
\renewcommand\paragraph{\@startsection{paragraph}{4}{\z@}{.9\linespacing\@plus.3\linespacing\@minus.2\linespacing}{-.5em}{\normalfont\itshape}}
\makeatother

\usepackage{times}
\usepackage[utf8]{inputenc}
\usepackage[vietnamese,english]{babel}
\usepackage[displaymath, tightpage]{preview}
\usepackage[all, cmtip]{xy}

\usepackage{amsmath, bm}
\usepackage{amsxtra}
\usepackage{amscd}
\usepackage{amsthm}
\usepackage{amsfonts}
\usepackage{amssymb}
\usepackage{enumitem}
\usepackage{array} 
\usepackage{placeins} 
\usepackage{bbm}
\usepackage{mathtools}
\usepackage[bitstream-charter]{mathdesign}
\usepackage[T1]{fontenc}

\usepackage[style=numeric, backend=biber, maxbibnames=99]{biblatex}
\AtBeginBibliography{\sloppy}

\usepackage{tikz}
\usepackage{tikz-cd}
\usetikzlibrary{matrix,calc,arrows}
\usepackage[new]{old-arrows}

\usepackage[hidelinks]{hyperref}
\usepackage{cleveref}

\usepackage[scaled=0.8]{helvet} 

\theoremstyle{plain}
\newtheorem{theorem}{Theorem}[section]

\newtheorem{proposition}[theorem]{Proposition}
\newtheorem{corollary}[theorem]{Corollary}

\theoremstyle{definition}
\newtheorem{remark}[theorem]{Remark}
\newtheorem{definition}[theorem]{Definition}

\AddToHook{env/lemma/begin}{\crefalias{theorem}{lemma}}
\AddToHook{env/proposition/begin}{\crefalias{theorem}{proposition}}
\AddToHook{env/corollary/begin}{\crefalias{theorem}{corollary}}
\AddToHook{env/remark/begin}{\crefalias{theorem}{remark}}
\AddToHook{env/definition/begin}{\crefalias{theorem}{definition}}
\AddToHook{env/example/begin}{\crefalias{theorem}{example}}

\DeclareMathOperator{\PSL}{PSL}

\newcommand\dqd{\foreignlanguage{vietnamese}{Đinh Quý Dương}}
\newcommand\BB{Białynicki-Birula }

\title[Baker--Akhiezer sections and Landau levels]{Baker--Akhiezer Sections from Higgs bundles\\
and Landau Levels in Hyperbolic Band Theory}
\author{\dqd}
\address{Department of Mathematics, University of Pennsylvania, Philadelphia, PA 19104, USA, and Department of Education, University of Khánh Hoà, Khánh Hoà, Vi\d{\^e}t Nam}
\email{duongdinh.mp@gmail.com}
\author{Steven Rayan}
\address{Centre for Quantum Topology and Its Applications (quanTA), Department of Mathematics and Statistics, University of Saskatchewan, Saskatoon, Saskatchewan S7N 5E6, Canada}
\email{rayan@math.usask.ca}
\date{}
\begin{document}
\begin{abstract}
As observed by Kienzle and the second author, rank-two $\mathcal{O}(3)$-twisted Higgs bundles on the projective line encode all hyperbolic crystals of genus two: the spectral curves are the unit cells, and the Hitchin fibers are the abelian Brillouin zones.
In a uniform magnetic field of two flux quanta, the lowest Landau level is one-dimensional at every nonzero momentum, and, as found by Iengo and Li, its dimension jumps to two at the zero momentum.
We identify the states of the lowest Landau level with the Baker--Akhiezer sections of the Higgs bundles, and their zeros, the vortices, with the Baker--Akhiezer divisors.
A family of nonzero momenta limits to the zero momentum exactly when each of the two vortices approaches the image of the other under the spatial inversion, and the states of the jumped level are labeled by the directions of limiting.
On the Higgs bundle side, the zero momentum is the point where the Hitchin fiber intersects the Hitchin section.
Such a family of Higgs bundles limits to the Hitchin section, thereby crossing to a lower stratum of the \BB stratification,
an example of the phenomenon investigated in terms of Baker--Akhiezer divisors by the first author with Figiel and Teschner.
\end{abstract}
\maketitle

\section{Introduction}
\subsection{Hyperbolic crystals and Higgs bundles}
Bloch's theorem states that the stationary states of an electron in a crystal can be chosen as eigenfunctions of the translations of the lattice \cite{AM76, Kuc16}.
The eigenvalues of such a state define a character of the lattice, the crystal momentum of the state.
The crystal momenta form a torus, the Brillouin zone.
Maciejko and the second author extended band theory to hyperbolic band theory, the band theory of lattices in the hyperbolic plane \cite{MR21}, which has since been developed in \cite{MR22, BGKMRT22, KR22, NR24, Ray26, MSTR26} (see also \cite{Ray26b}).
In this setting, the unit cells of hyperbolic lattices are compact Riemann surfaces $\Sigma$ of genus $g\geq2$.
The abelian crystal momenta, i.e.\ the characters of the translation group $\pi_1(\Sigma)$ in $U(1)$, form a torus of real dimension $2g$ \cite{MR21}, the Jacobian of the unit cell \cite{KR22}.

There have been experimental work realized for this theoretical direction of research.
Hyperbolic lattices have been realized in circuit quantum electrodynamics \cite{KFH19}, electric circuits \cite{Len22, ZYSSZ22, Che23, ZDZSZ23}, photonic ring resonators \cite{Hua24} and superconducting resonator networks \cite{XMGTRM25}.
In electric circuits, tunable complex phases can emulate any momentum of a crystal whose unit cell has genus two, the case of this paper \cite{Che23}.

As observed by Kienzle and the second author \cite{KR22}, rank-two $\mathcal{O}(3)$-twisted Higgs bundles on the Riemann sphere $C=\mathbb{P}^1$ encode all hyperbolic crystals of genus two together with their abelian crystal momenta.
Such a Higgs bundle is a pair $(E,\phi)$ of a vector bundle $E$ of rank two on $C$ and a Higgs field $\phi\colon E\to E\otimes\mathcal{O}(3)$.
For $\phi$ trace-free and $\det E\cong\mathcal{O}(-3)$, the stable Higgs bundles form a smooth moduli space $\mathcal{M}$ of dimension $9$ (\Cref{sec-crystals}).
The Hitchin map is $h\colon\mathcal{M}\to B$, $(E,\phi)\mapsto\det\phi$, where the Hitchin base $B=H^0(\mathcal{O}(6))$ is the space of binary sextics.

A sextic $b\in B$ with simple zeros defines the spectral curve
\[ \widetilde{C}_b=\{\lambda^2+b=0\}\subset\mathrm{Tot}(\mathcal{O}(3)) \]
of the eigenvalues $\lambda$ of the Higgs fields with $\det\phi=b$.
The spectral curve has genus two and is the unit cell of a hyperbolic crystal \cite{KR22}.
For instance, the spectral curve of $b=x-x^5$ is the Bolza curve, the unit cell of the lattice $\{8,8\}$ \cite{MR21}.

In the Higgs bundle framework, a crystal momentum is a line bundle $\mathcal{L}$ of degree $0$ on $\widetilde{C}_b$.
The spectral correspondence \cite{BNR89} identifies $\mathcal{L}$ with the Higgs bundle $(E,\phi)=(\pi_\ast\mathcal{L},\pi_\ast\lambda)$ in $h^{-1}(b)$, where $\pi\colon\widetilde{C}_b\to C$ is the projection.
So the Hitchin fiber $h^{-1}(b)$ is the abelian Brillouin zone $\mathrm{Jac}(\widetilde{C}_b)$ (\Cref{tab-dict}).

\subsection{Lowest Landau levels and Baker--Akhiezer divisors}
In a uniform magnetic field, the spectrum of a charged particle on the hyperbolic plane consists of finitely many Landau levels below a continuous spectrum \cite{Com87}.
Magnetic fields on hyperbolic lattices have been studied in \cite{IAM21, SUTB22}.
On a unit cell, a magnetic field $\mathbf{B}$ with an integral number of flux quanta determines a Hermitian line bundle $L_{\mathbf{B}}$ with a unitary connection of curvature $-i\mathbf{B}$, up to the tensor product with a flat unitary line bundle.
The degree of $L_{\mathbf{B}}$ is the number of flux quanta.
The lowest Landau level of a momentum $\mathcal{L}$ consists of the holomorphic sections of $L_{\mathbf{B}}\otimes\mathcal{L}$ \cite{IL94}.
The wave function of a lowest Landau level state vanishes at isolated points, its \emph{vortices}, around which its phase winds.
A state has one vortex per flux quantum, and the positions of its vortices determine it up to scale.
Filling the lowest Landau level produces the integer quantum Hall effect, whose quantized Hall conductance is, for crystals, a Chern number of the filled bands \cite{TKNN82}.

The main theme in this paper is to interpret, for $n\geq1$, the space $H^0(\mathcal{L}\otimes K_{\widetilde{C}_b}^n)$ as the lowest Landau level of momentum $\mathcal{L}$ in the uniform magnetic field $\mathbf{B}=n\,\mathrm{vol}$ of $2n$ flux quanta.
Here $K_{\widetilde{C}_b}$ is the canonical line bundle of $\widetilde{C}_b$, and $\mathrm{vol}$ is the area form of the hyperbolic metric of curvature $-1$ on the unit cell (\Cref{sec-fields}).
For $n\geq2$, the flux exceeds $2g-2=2$, so the lowest Landau level has dimension $2n-1$ at every momentum \cite{IL94}.
For $n=1$, the flux is the borderline one: the lowest Landau level is one-dimensional at every nonzero momentum, and its dimension jumps to two at the zero momentum $\Gamma_b=\mathcal{O}_{\widetilde{C}_b}$.
Iengo and Li found this jump by counting dimensions with the Riemann--Roch theorem \cite{IL94}.
Studies of the Hall conductance as an adiabatic curvature on closed Riemann surfaces assume a flux above $2g-2$, so that the lowest Landau level has the same dimension at every momentum \cite{ASZ94, KMMW17}.

Physically, the jump is a band touching (\Cref{fig-band-touching}).
For a free electron in the magnetic field $\mathbf{B}=\mathrm{vol}$, the band above the lowest Landau level is the lowest band at zero magnetic field shifted by $1$ (\Cref{sec-fields}, cf.\ \cite{IL94, CL25}).
The lowest band at zero magnetic field vanishes exactly at $\Gamma_b$, where its state is the constant wave function.
So in the magnetic field $\mathbf{B}=\mathrm{vol}$, the band above the lowest Landau level touches it exactly at $\Gamma_b$,
thereby adding a state of this band to the lowest Landau level.

\begin{figure}[ht]
\centering
\includegraphics{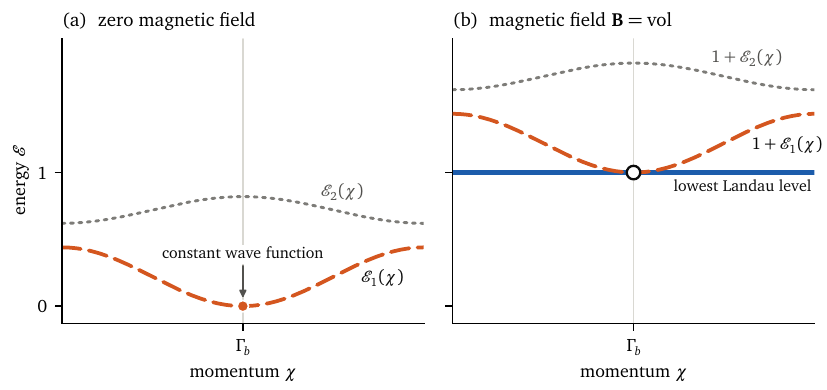}
\caption{The jump of the lowest Landau level at $\Gamma_b$ is a band touching. The bands, drawn schematically, are those of a free electron in a crystal with a unit cell of genus two, along a loop of momenta through $\Gamma_b$. (a) At zero magnetic field, the lowest band $\mathcal{E}_1$ vanishes exactly at $\Gamma_b$. (b) In the magnetic field $\mathbf{B}=\mathrm{vol}$, the bands of (a) are shifted by $1$, so the band $1+\mathcal{E}_1$ touches the lowest Landau level at $\Gamma_b$, where the level jumps from dimension $1$ to $2$.}
\label{fig-band-touching}
\end{figure}

On the Higgs bundle side, 
adjunction identifies the Landau level $H^0(\widetilde{C}_b, \mathcal{L}\otimes K_{\widetilde{C}_b}^n)$ with the space of morphisms from $\mathcal{O}(-n)$ to $E$,
\[ \mathrm{Hom}(\mathcal{O}(-n),E)\xrightarrow{\ \sim\ }H^0(\widetilde{C}_b,\mathcal{L}\otimes K_{\widetilde{C}_b}^n), \qquad i\longmapsto\widetilde{s}_i(\phi). \]
Following the work of Teschner with the first author \cite{DT23}, 
we call $\widetilde{s}_i(\phi)$ the \emph{Baker--Akhiezer section} of the triple $(\mathcal{O}(-n)\xrightarrow{i}E,\phi)$, after the Baker--Akhiezer functions of integrable systems \cite{Kri77, Dub81}.
We call the zero divisor $\widetilde{D}$ of $\widetilde{s}_i(\phi)$ the \emph{Baker--Akhiezer divisor}.%
\footnote{
The analogues for flat connections on Riemann surfaces of the points of the Baker--Akhiezer divisor are the apparent singularities of the second-order differential equation defined by a line subbundle, together with their residue parameters \cite{DT23, DFT26}.
Such a line subbundle is a generic oper structure, which every irreducible connection admits \cite{Ari16}.
These singularities arise in the classical limit of conformal field theory from insertions of degenerate fields \cite{Res92, Tes17}, and play a central role in the geometric Langlands correspondence \cite{DT23}, which has been proved by Gaitsgory--Raskin et al.\ \cite{GLC}.
}
Such divisors already appeared in the original work of Hitchin introducing the moduli of Higgs bundles \cite{Hit87a},
and more recently in the work of Hausel--Hitchin \cite{HH21}.
For $n\geq2$, the Baker--Akhiezer divisors move in the linear systems $|\mathcal{L}\otimes K_{\widetilde{C}_b}^n|$, of dimension $2n-2$ (\Cref{tab-fields}).
In band theory, the Baker--Akhiezer section is a lowest Landau level state, whose vortices are the points of the Baker--Akhiezer divisor.

The recent joint work of the first author with Figiel and Teschner 
investigates families of Higgs bundles in which two points in the norm of Baker--Akhiezer divisors collide and disappear,
indicating a crossing to a lower stratum in the moduli space \cite{DFT26}.
On spectral curves, this phenomenon correspond to two points of the Baker--Akhiezer divisor limiting to a pair exchanged by the spectral involution \cite{DFT26}.
In this paper, we show that this degeneration corresponds to a family of Higgs bundles in one fixed Hitchin fiber, 
i.e. one Brillouin zone,
that limits to the origin of the Jacobian of spectral curves $\widetilde{C}_b$,
i.e. the zero momentum $\Gamma_b$, where the lowest Landau level jumps (\Cref{prop-crossing}(ii)).

\subsection{Main results}
We now summarize the main results of this paper
in both the Higgs bundle framework and the hyperbolic band theory interpretation.
Throughout this paper, we consider smooth hyperelliptic spectral curves
\[ \widetilde{C}_b \xrightarrow{\pi} C = \mathbb{P}^1. \]

\paragraph{Lowest Landau levels and vortices}
For a nonzero momentum $\mathcal{L} \in \mathrm{Jac}(\widetilde{C}_b)$, 
the bundle $E = \pi_\ast \mathcal{L}$ splits as $\mathcal{O}(-1)\oplus\mathcal{O}(-2)$ (\Cref{sec-crystals}), so $H^0(\mathcal{L}\otimes K_{\widetilde{C}_b})$ is one-dimensional, spanned by the Baker--Akhiezer section of the inclusion of the summand $\mathcal{O}(-1)$.
Concretely, if $\phi=\left(\begin{smallmatrix}\phi_0&\phi_-\\\phi_+&-\phi_0\end{smallmatrix}\right)$ in local frames adapted to $\mathcal{O}(-1)\subset E$, the Baker--Akhiezer divisor $\widetilde{D}=\widetilde{p}_1+\widetilde{p}_2$ lies on the sheet $\lambda=-\phi_0$ over the zeros of $\phi_+$ (\Cref{rem-BA-concrete}).
The Baker--Akhiezer divisor then determines the momentum via 
\[ \mathcal{L}\cong K_{\widetilde{C}_b}^{-1}(\widetilde{D}). \]

In band theory, the lowest Landau level of a nonzero momentum in the magnetic field $\mathbf{B}=\mathrm{vol}$ is spanned by the Baker--Akhiezer section.
So the two vortices of this state determine the momentum.

\paragraph{The zero momentum}
At the zero momentum $\mathcal{L} \cong \mathcal{O}_{\widetilde{C}_b} \equiv 0 \in \mathrm{Jac}(\widetilde{C}_b)$, 
the corresponding Higgs bundle is the intersection $\mathrm{Hit}(b)$ of the Hitchin section with the Hitchin fiber $h^{-1}(b)$,
\[ \mathrm{Hit}(b)=\left(\mathcal{O}\oplus\mathcal{O}(-3),\begin{pmatrix} 0 & -b \\ 1 & 0 \end{pmatrix}\right). \]
The corresponding Baker--Akhiezer divisors of $\mathrm{Hit}(b)$ 
induced by the morphism 
\[ \mathcal{O}(-1) \xrightarrow{s_p} \mathcal{O} \hookrightarrow \mathcal{O} \oplus \mathcal{O}(-3) \]
is of the form $\pi^\ast p$.
These divisors form the exceptional curve  
$|K_{\widetilde{C}_b}| \cong |\pi^\ast \mathcal{O}(1)| \cong\mathbb{P}^1$ of the map
\[ \widetilde{C}_b^{(2)}\longrightarrow h^{-1}(b), \qquad \widetilde{D}\longmapsto K_{\widetilde{C}_b}^{-1}(\widetilde{D}), \]
which is the blow-up of $h^{-1}(b)$ at $\Gamma_b$.%
\footnote{This is an example of the shaky equal wobbly phenomenon, investigated in \cite{DDNP26}.}
A family of Higgs bundles in $h^{-1}(b)\smallsetminus\{\mathrm{Hit}(b)\}$ limits to $\mathrm{Hit}(b)$ exactly when the Baker--Akhiezer divisors limit to the curve $|K_{\widetilde{C}_b}|$ of the fibers $\pi^\ast p$ (\Cref{prop-crossing}(ii)).
If the Baker--Akhiezer divisors limit to one fiber $\pi^\ast p$, then $p$ is the direction of approach.

The translation to band theory language is that 
the jump of the lowest Landau level happens where the Brillouin zone meets the Hitchin section.
A family of nonzero momenta limits to $\Gamma_b$ exactly when each of the two vortices approaches the image of the other under the spatial inversion, the hyperelliptic involution $\lambda\mapsto-\lambda$ of $\widetilde{C}_b$.
If a family of nonzero momenta limits to $\Gamma_b$ in the complex direction $p$, its lowest Landau level states limit, up to scaling, to the state $\pi^\ast s_p$ (\Cref{sec-LLL}, \Cref{fig-vortex-pairs}).
In other words, the states of the jumped level are labeled by the directions of limiting.

\begin{figure}[ht]
\centering
\includegraphics{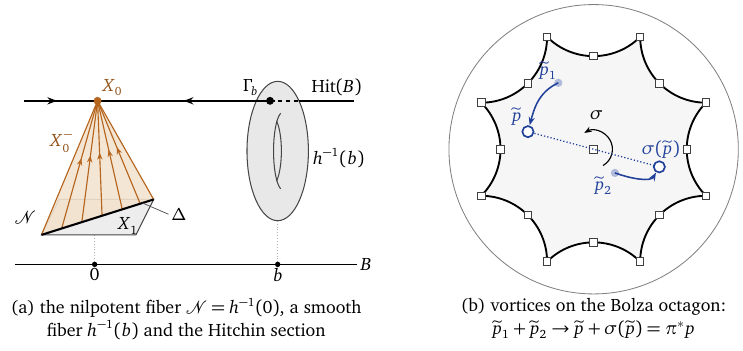}
\caption{(a) The nilpotent fiber $\mathcal{N}=h^{-1}(0)$ and a smooth fiber $h^{-1}(b)$. The Hitchin section, which is the upward flow of $X_0=\{\mathrm{Hit}(0)\}$, intersects $h^{-1}(b)$ in $\Gamma_b=\mathrm{Hit}(b)$. 
The closure of the downward flow $X_0^-$ intersects $X_1\cong C^{(2)}$ along the wobbly locus--the diagonal $\Delta$. (b) For $b=x-x^5$, the unit cell is the Bolza octagon, with the inversion $\sigma\colon z\mapsto-z$. The family of nonzero momenta limiting to $\Gamma_b$ in the complex direction $p$ has vortices $\widetilde{p}_1$ and $\widetilde{p}_2$ limiting to $\pi^\ast p=\widetilde{p}+\sigma(\widetilde{p})$.}
\label{fig-vortex-pairs}
\end{figure}

\paragraph{Degeneration to the nilpotent fiber}
Under the $\mathbb{C}^\ast$-action $t.(E,\phi)=(E,t\phi)$ on $\mathcal{M}$, as $t\to0$ 
every Higgs bundle flows to a fixed point in the nilpotent fiber $\mathcal{N}=h^{-1}(0)$ (\Cref{sec-BBdec}).
These fixed points are called \emph{Hodge bundles} and satisfy the Griffiths transversality condition, 
i.e.\ $(L_1\oplus L_2,\phi)$ with $\phi(L_1)\subset L_2\otimes\mathcal{O}(3)$ and $\phi(L_2)=0$.
In our setting, Hodge bundles are either the uniformizing Higgs bundle $\mathrm{Hit}(0)$ or of the form
\[ \left(\mathcal{O}(-1)\oplus\mathcal{O}(-2),\begin{pmatrix} 0 & 0 \\ s_D & 0 \end{pmatrix}\right), \]
where $s_D\in H^0(\mathcal{O}(2))$ vanishes on a divisor $D\in C^{(2)}$.
As $D$ varies, these Hodge bundles form a component $X_1\cong C^{(2)}$ of the fixed point locus (\Cref{prop-fixed}).
A Higgs bundle in $h^{-1}(b)\smallsetminus\{\mathrm{Hit}(b)\}$ with Baker--Akhiezer divisor $\widetilde{D}$ flows to the Hodge bundle with $D=\mathrm{Nm}(\widetilde{D})=\pi(\widetilde{p}_1)+\pi(\widetilde{p}_2)$ (\Cref{lem-BA}).

Following Donagi--Pantev and Hausel--Hitchin \cite{DP09, HH21}, 
we say that a Hodge bundle is \emph{very stable} if no other nilpotent Higgs bundle flows to it, and \emph{wobbly} otherwise.
The nilpotent fiber is the union $\mathcal{N}=X_1\cup\overline{X_0^-}$ of two surfaces intersecting along the diagonal of $X_1$, where $X_0^-$ consists of the Higgs bundles $x$ with $\lim_{t\to\infty}t.x=\mathrm{Hit}(0)$ (\Cref{prop-nilpotent-component}, \Cref{fig-vortex-pairs}(a)).
As for $K$-twisted Higgs bundles on curves of genus $g>1$ \cite[Thm.~1.2]{HH21},
the Hodge bundle with divisor $D$ is wobbly exactly when $D$ is not reduced,
i.e.\ when $D$ lies on the diagonal $\Delta \subset C^{(2)} \cong X_1$ (\Cref{cor-wobbly}).
For $D=p_1+p_2$ with $p_1\neq p_2$, the upward flow of the Hodge bundle, i.e.\ the set of Higgs bundles that flow to it, intersects a generic Hitchin fiber in the four Higgs bundles with Baker--Akhiezer divisors $\widetilde{p}_1+\widetilde{p}_2$, $\widetilde{p}_k\in\pi^{-1}(p_k)$ (\Cref{lem-BA}).

In band theory, the $\mathbb{C}^\ast$-action changes neither the crystal nor the momentum, but the limit remembers only the positions of the two vortices up to the inversion.
A very stable Hodge bundle in $X_1$ is the limit of four momenta of a generic crystal, one for each way of placing the two vortices over the two positions.
At a wobbly Hodge bundle, the two positions collide.

\paragraph{Upward flows, theta functions and a four-dimensional crystal}
For $K$-twisted Higgs bundles, Hausel--Hitchin \cite{HH21} study the Fourier--Mukai transforms of the structure sheaves of upward flows,
and relate them to mirror symmetry and  
the approach of Donagi--Pantev to the geometric Langlands correspondence \cite{DP09}.
Let $J=h^{-1}(b)$, and let $\alpha\in X_1$ be a very stable Hodge bundle whose upward flow intersects $J$ in four points.
The Fourier--Mukai transform of the structure sheaf of these four points pulls back along the Abel--Jacobi map $\widetilde{C}_b\to J$ to the sum of the four spectral line bundles (\Cref{sec-theta}).
At a nonzero momentum $\mathcal{L}$, the Baker--Akhiezer section is up to scaling the pullback along the Abel--Jacobi map of a translate $\vartheta_{\mathcal{L}}$ of the theta function (\Cref{sec-theta}).

In band theory language, the pullback to the unit cell 
of the transform is the sum of the momentum sectors of the four momenta of the upward flow.

The Brillouin zone $J$, a torus of real dimension four, is also the Brillouin zone of a four-dimensional Euclidean crystal, whose unit cell is the dual torus $\widehat J$ \cite{KR22}.
Whereas the unit cell $\widetilde{C}_b$ of the hyperbolic crystal is curved and of real dimension two, the unit cell $\widehat J$ of this crystal is flat and of real dimension four.
When $\widehat J$ is identified with $J$ by the principal polarization, the Abel--Jacobi map embeds the hyperbolic unit cell in the four-dimensional one as the theta divisor $\Theta$.
In the magnetic field of the principal polarization, the lowest Landau level of the four-dimensional crystal is one-dimensional at every momentum $\mathcal{L}$, including $\Gamma_b$, and is spanned by $\vartheta_{\mathcal{L}}$.
Pulling back along the Abel--Jacobi map sends these lowest Landau level states of the Euclidean crystal to those of the hyperbolic crystal in the magnetic field $\mathbf{B}=\mathrm{vol}$.
The four-dimensional crystal suggests answers to the questions of \cite{KR22} on the crystallographic meaning of the principal polarization, the theta divisor and theta functions (\Cref{sec-4d}, \Cref{tab-crystals}).

\begin{table}[ht]
\centering
\renewcommand{\arraystretch}{1.5}
\begin{tabular}{>{\raggedright\arraybackslash}p{0.45\textwidth}|>{\raggedright\arraybackslash}p{0.45\textwidth}}
\textbf{Higgs bundles} & \textbf{Band theory} \\ \hline
generic point $b$ of the Hitchin base, spectral curve $\widetilde{C}_b$ & crystal with unit cell $\widetilde{C}_b$ \\
Hitchin fiber $h^{-1}(b)=\mathrm{Jac}(\widetilde{C}_b)$ & abelian Brillouin zone \\
spectral line bundle $\mathcal{L}=\mathcal{L}_{(E,\phi)}$ & crystal momentum $\chi$, with $\mathcal{L}=(L_\chi,\nabla_\chi^{0,1})$ \\
$L^2$ sections of $\mathcal{L}$ & wave functions of momentum $\mathcal{L}$ \\ \hline
point $\mathrm{Hit}(b)$ of the Hitchin section, with $\mathcal{L}=\mathcal{O}_{\widetilde{C}_b}$ (\Cref{sec-crystals}) & zero momentum $\Gamma_b$, cf.\ \cite{KR22}, where, in the magnetic field $\mathbf{B}=\mathrm{vol}$, the next band touches the lowest Landau level \\
$K_{\widetilde{C}_b}$, with the hyperbolic metric & uniform magnetic field $\mathbf{B}=\mathrm{vol}$, of $2$ flux quanta \\
$H^0(\mathcal{L}\otimes K_{\widetilde{C}_b})$ & lowest Landau level of momentum $\mathcal{L}$ \\
Baker--Akhiezer section $\widetilde{s}_i(\phi)$ and divisor $\widetilde{D}=\widetilde{p}_1+\widetilde{p}_2$ & lowest Landau level state and its two vortices \\
exceptional curve $|K_{\widetilde{C}_b}|\cong\mathbb{P}(T_{\mathrm{Hit}(b)}h^{-1}(b))$ of the blow-up & directions of approach to $\Gamma_b$, labeling the states $\pi^\ast s_p$ of the jumped level $H^0(K_{\widetilde{C}_b})$ \\
$\underset{t\to0}{\lim}t.(E,\phi)=\mathrm{Nm}(\widetilde{D})\in X_1\cong C^{(2)}$, for $\mathcal{L}\neq\mathcal{O}_{\widetilde{C}_b}$ & positions of the two vortices up to inversion \\
the four points in $h^{-1}(b)$ of the upward flow of a very stable Hodge bundle & the four momenta with the same vortex positions up to inversion \\
wobbly Hodge bundles, with divisor $D=2p$ & collision of the two vortex positions
\end{tabular}
\vspace{8pt} 
\caption{The dictionary between Higgs bundles and hyperbolic band theory: above the line, that of \cite{KR22}; below the line, its extension to magnetic fields and lowest Landau levels \cite{IL94} (\Cref{sec-degen,sec-BA}).}
\label{tab-dict}
\end{table}

\paragraph{Organization}
\Cref{sec-band} recalls band theory, magnetic fields, lowest Landau levels and the dictionary of \cite{KR22}.
\Cref{sec-BB} introduces the Baker--Akhiezer divisors, decomposes $\mathcal{M}$ into the upward flows of the fixed points (the \BB decomposition) and computes the nilpotent fiber.
\Cref{sec-degen} describes how the Baker--Akhiezer divisors degenerate under the $\mathbb{C}^\ast$-action and as the momentum limits to $\Gamma_b$.
\Cref{sec-BA} relates the Baker--Akhiezer sections to lowest Landau levels, to theta functions and to a four-dimensional crystal.

\section{Band theory and spectral curves}\label{sec-band}
\subsection{Preliminaries on band theory}
We recall the notions of band theory used in this paper.
Standard references are \cite{AM76, Kuc16}.
We use units in which the reduced Planck constant $\hbar$, the charge of the particle and twice its mass are equal to $1$.

\paragraph{Crystals and Bloch's theorem}
An electron in a crystal in $\mathbb{R}^d$ is described by the Hamiltonian
\[ H=-\Delta+V \quad\text{on } L^2(\mathbb{R}^d), \]
where the potential $V$ is invariant under a lattice $\Lambda\cong\mathbb{Z}^d$ of translations.
The translations $\psi(x)\mapsto\psi(x+\gamma)$, $\gamma\in\Lambda$, commute with $H$ and with each other.
Bloch's theorem states that the solutions of the eigenvalue equation $H\psi=\mathcal{E}\psi$ on $\mathbb{R}^d$ can be chosen to be \emph{Bloch functions}, i.e.\ functions $\psi_\chi$ with
\[ \psi_\chi(x+\gamma)=\chi(\gamma)\psi_\chi(x), \qquad \gamma\in\Lambda, \]
for a character $\chi\colon\Lambda\to U(1)$ \cite{AM76}.

Every character of $\Lambda$ is of the form
\[ \chi(\gamma)=e^{i\langle k,\gamma\rangle}, \qquad k\in(\mathbb{R}^d)^\ast. \]
Two covectors $k$ and $k'$ give the same character if and only if $k-k'\in2\pi\Lambda^\ast$, where $\Lambda^\ast=\{\ell\in(\mathbb{R}^d)^\ast:\ell(\Lambda)\subset\mathbb{Z}\}$ is the dual lattice.
The class of $k$ modulo the \emph{reciprocal lattice} $2\pi\Lambda^\ast$ is the \emph{crystal momentum}.
The crystal momenta form the \emph{Brillouin zone}
\[ \mathrm{BZ}=(\mathbb{R}^d)^\ast/2\pi\Lambda^\ast\cong\mathrm{Hom}(\Lambda,U(1)), \]
a torus of dimension $d$.
We also speak of the momentum $\chi$.
The zero momentum $\chi=1$ is denoted $\Gamma$.

The positions up to lattice translations form the torus $T=\mathbb{R}^d/\Lambda$, the \emph{unit cell} of the crystal.
Physicists often call a fundamental domain of $\Lambda$ the unit cell, but we use the term for the torus $T$, which is obtained by identifying the faces of a fundamental domain.
We do the same for hyperbolic crystals.
The Brillouin zone is the dual torus of $T$: as $\Lambda=\pi_1(T)$, it parametrizes the flat unitary line bundles on $T$.

\paragraph{Momentum sectors}
For $\chi\in\mathrm{BZ}$, let $L_\chi$ be the line bundle on $T$ whose total space is the quotient of the trivial line bundle $\mathbb{R}^d\times\mathbb{C}$ by the action of $\Lambda$ given by
\[ L_\chi=(\mathbb{R}^d\times\mathbb{C})/\Lambda, \qquad \gamma.(x,c)=(x+\gamma,\chi(\gamma)c). \]
A section of $L_\chi$ over $T$ is the same as a Bloch function with character $\chi$.
Since $\Lambda$ acts on the fibers by the constants $\chi(\gamma)$, the trivial connection $d$ of $\mathbb{R}^d\times\mathbb{C}$ descends to a flat connection $\nabla_\chi$ on $L_\chi$.
The connection $\nabla_\chi$ is unitary as $|\chi(\gamma)|=1$.
The Hamiltonian $H$ therefore induces the operator
\[ H_\chi=\nabla_\chi^\ast\nabla_\chi+V \quad\text{on } L^2(T,L_\chi). \]
We call the flat line bundle $L_\chi$ the \emph{momentum sector} of $\chi$.

\paragraph{Bands}
Since $T$ is compact, $H_\chi$ has discrete spectrum.
We write the eigenvalues of $H_\chi$, counted with multiplicity, and the eigenfunctions as
\[ H_\chi\psi_{n,\chi}=\mathcal{E}_n(\chi)\psi_{n,\chi}, \qquad \mathcal{E}_1(\chi)\leq\mathcal{E}_2(\chi)\leq\cdots, \]
labeled by the \emph{band index} $n\geq1$ and the momentum $\chi$.
For fixed $n$, the eigenvalue $\mathcal{E}_n(\chi)$ depends continuously on $\chi\in\mathrm{BZ}$.
The function $\mathcal{E}_n$ on $\mathrm{BZ}$ is the $n$-th \emph{band} \cite{Kuc16}.
When $\mathcal{E}_n(\chi)$ is a simple eigenvalue at every $\chi$, we also call the line bundle of the eigenlines $\mathbb{C}\psi_{n,\chi}$ over $\mathrm{BZ}$ the band.

\paragraph{Hyperbolic crystals}
Hyperbolic band theory is the band theory of lattices in the hyperbolic plane $\mathbb{H}$.
An example is the $\{8,8\}$ lattice, the tessellation of $\mathbb{H}$ by regular octagons, eight of them meeting at each vertex \cite{MR21}.
The translations of a hyperbolic lattice form a Fuchsian group $G\subset\PSL(2,\mathbb{R})$ acting freely and cocompactly on $\mathbb{H}$.
The unit cell is the compact Riemann surface $\Sigma=\mathbb{H}/G$ of some genus $g\geq2$, with the hyperbolic metric \cite{BGKMRT22}.

The momentum sectors of the crystal are labeled by the irreducible unitary representations of the group $G\cong\pi_1(\Sigma)$ \cite{MR22, KR22}.
Since $G$ is not abelian, some of these representations have dimension greater than one.
The decomposition of wave functions on $\mathbb{H}$ into momentum sectors is the hyperbolic Bloch transform \cite{NR24}.
The transform maps the wave functions to sections of flat unitary bundles on $\Sigma$ and intertwines the Laplacian of $\mathbb{H}$ with the covariant Laplacians of these bundles.
The sectors of the one-dimensional representations, the characters $\chi\colon G\to U(1)$, are called \emph{abelian}, whereas the others are called nonabelian \cite{MR21, KR22}.
We only use the abelian sectors.
A character factors through the abelianization $G/[G,G]=H_1(\Sigma,\mathbb{Z})$,
\[ \chi\in\mathrm{Hom}(G,U(1))=\mathrm{Hom}(H_1(\Sigma,\mathbb{Z}),U(1))\cong U(1)^{2g}, \]
so $\chi$ is determined by the phases picked up by the electron along the $2g$ cycles of a symplectic basis: the Aharonov--Bohm phases of magnetic fluxes threading the holes of the unit cell, with no magnetic field on the unit cell itself \cite{KR22}.
The abelian Bloch functions satisfy
\[ \psi_\chi(\gamma z)=\chi(\gamma)\psi_\chi(z), \qquad \gamma\in G. \]

Equivalently, the abelian Bloch functions with character $\chi$ are the sections of the flat unitary line bundle $L_\chi=(\mathbb{H}\times\mathbb{C})/G$ on $\Sigma$, where $\gamma.(z,c)=(\gamma z,\chi(\gamma)c)$.
We denote the flat connection of $L_\chi$ by $\nabla_\chi$.
The Hamiltonian $H$ induces the operator
\[ H_\chi=\nabla_\chi^\ast\nabla_\chi+V \quad\text{on } L^2(\Sigma,L_\chi), \]
where the adjoint is taken with respect to the hyperbolic metric, and the potential $V$ is a function on $\Sigma$ \cite{KR22}.
The holomorphic line bundle $\mathcal{L}=(L_\chi,\nabla_\chi^{0,1})$ has degree $0$, as $\nabla_\chi$ is flat.
Recall from the rank-one case of the theorem of Narasimhan--Seshadri that a holomorphic line bundle of degree $0$ carries a unique flat unitary connection whose $(0,1)$-part is its holomorphic structure.
It follows that $\chi\mapsto\mathcal{L}$ is a bijection from $\mathrm{Hom}(G,U(1))$ onto $\mathrm{Jac}(\Sigma)$.
So the abelian Brillouin zone is the Jacobian \cite{KR22}.
Since $\chi$ is recovered from $\mathcal{L}$ as the holonomy of its flat unitary connection, no information is lost by labeling the momentum by $\mathcal{L}$.

\paragraph{Genus-two crystals}
The unit cells of the crystals in this paper have genus two.
Identifying the opposite sides of one octagon of the $\{8,8\}$ lattice gives the unit cell of the lattice, the Bolza curve \cite{MR21}
\[ y^2=x^5-x. \]
The rotation by $\pi$ about the center of the octagon is a symmetry of the crystal, the spatial inversion.
The quotient map by the inversion is the double cover $x\colon\Sigma\to\mathbb{P}^1$, branched at the images $0,\pm1,\pm i,\infty$ of the six fixed points of the inversion \cite{KR22}.
In the notation of the introduction, $\Sigma\cong\widetilde{C}_b$ for $b=x-x^5$, with $y=\lambda$.
We keep the letter $\Sigma$ for the unit cell of a crystal of any genus, and write $\widetilde{C}_b$ when the unit cell is viewed as a spectral curve.

Every compact Riemann surface of genus two is hyperelliptic: it is a double cover of $\mathbb{P}^1$ branched at six points, with an equation expressing $y^2$ as a polynomial in $x$ of degree six, or five if $\infty$ is a branch point.
Every such curve is the unit cell of a crystal with inversion symmetry \cite{KR22}.

\subsection{Magnetic fields and Landau levels}\label{sec-fields}
A magnetic field on a Riemannian surface $Y$ is a closed $2$-form $\mathbf{B}$, which enters the Hamiltonian through the minimal coupling
\[ \nabla=d-i\mathbf{A}, \qquad d\mathbf{A}=\mathbf{B}, \qquad H=\nabla^\ast\nabla+V. \]
Globally, the wave functions are the sections of a Hermitian line bundle $L$ on $Y$ with a unitary connection $\nabla$ of curvature $-i\mathbf{B}$.
In a local unitary frame, $\nabla=d-i\mathbf{A}$.
If $Y$ is closed, the number of \emph{flux quanta}
\[ N_\Phi=\frac{1}{2\pi}\int_{Y}\mathbf{B}=\deg L \]
is an integer  by Chern--Weil theory.
In physics, this integrality is Dirac's quantization of the magnetic flux.

Two unitary connections with the same curvature differ by a closed $1$-form, so the magnetic field determines $(L,\nabla)$ only up to the tensor product with a flat unitary line bundle $L_\chi$, whose holonomies $\chi$ are the crystal momentum.
We fix one such pair, $(L_{\mathbf{B}},\nabla_{\mathbf{B}})$, the \emph{line bundle of the magnetic field}.
The wave functions of momentum $\chi$ in the magnetic field are then the sections of the tensor product of Hermitian line bundles with unitary connections
\[ L=L_{\mathbf{B}}\otimes L_\chi, \qquad \nabla=\nabla_{\mathbf{B}}\otimes1+1\otimes\nabla_\chi. \]
The bundle $L$ varies with $\chi$, whereas $L_{\mathbf{B}}$ does not.
At zero magnetic field, $L_{\mathbf{B}}=\mathcal{O}$, so $L=L_\chi$ is the momentum sector of $\chi$, cf.\ \cite{KR22}.

\paragraph{Landau levels}
If the magnetic field is uniform, i.e.\ $\mathbf{B}=B_0\,\mathrm{vol}$ for a constant $B_0>0$, the Hamiltonian $\nabla^\ast\nabla$ without potential has explicit eigenvalues, the \emph{Landau levels}.
In the plane they are $B_0(2n+1)$, $n\geq0$, whereas in the hyperbolic plane $\mathbb{H}$ they are $B_0(2n+1)-n(n+1)$ for the integers $0\leq n<B_0-\frac12$, below the continuous spectrum $[B_0^2+\frac14,\infty)$ \cite{Com87}.
On a closed hyperbolic surface they are eigenvalues of finite multiplicity \cite{IL94, CL25}.
We only use the lowest one, of energy $B_0$.
On $\mathbb{H}$, this level exists if and only if $B_0>\frac12$.

The Weitzenb\"ock formula \cite{IL94, ASZ94}
\[ \nabla^\ast\nabla=2\bar\partial^\ast\bar\partial+B_0, \]
where $\bar\partial=\nabla^{0,1}$ is the $(0,1)$-part of $\nabla$, gives
\[ \langle\psi,\nabla^\ast\nabla\psi\rangle=2\|\bar\partial\psi\|^2+B_0\|\psi\|^2\geq B_0\|\psi\|^2, \]
with equality if and only if $\bar\partial\psi=0$.
On a Riemann surface the $(0,1)$-part of every connection is integrable, so $\bar\partial$ is a holomorphic structure on $L$.
Conversely, a Hermitian holomorphic line bundle determines its Chern connection, the unique unitary connection whose $(0,1)$-part is the holomorphic structure.
The holomorphic line bundle $(L,\nabla^{0,1})$ is the tensor product of the holomorphic line bundles $(L_{\mathbf{B}},\nabla_{\mathbf{B}}^{0,1})$ and $(L_\chi,\nabla_\chi^{0,1})$.
It follows from the Weitzenb\"ock formula that $B_0$ is an eigenvalue of $\nabla^\ast\nabla$ if and only if $L$ has a nonzero holomorphic section, and $B_0$ is then the lowest eigenvalue.
The eigenspace of $B_0$, the \emph{lowest Landau level} of momentum $\chi$, is the space of holomorphic sections \cite{IL94, Kle16}
\begin{equation}\label{eq-LLL}
H^0(Y,L_{\mathbf{B}}\otimes L_\chi).
\end{equation}
Otherwise every eigenvalue exceeds $B_0$, so the lowest Landau level is zero.

In a local holomorphic frame $e$ of $L$, a section $\psi=fe$ lies in the lowest Landau level if and only if the function $f$ is holomorphic.
For instance, for the magnetic field $\mathbf{B}=B_0\,dx\wedge dy$ on $\mathbb{C}$ and the frame with $\|e\|^2=e^{-B_0|z|^2/2}$, the sections $fe$ are the familiar wave functions $f(z)\,e^{-B_0|z|^2/4}$ of the lowest Landau level in the plane, written in a unitary frame.

The description \eqref{eq-LLL} of the lowest Landau level also holds on a K\"ahler manifold, for a magnetic field of type $(1,1)$ whose contraction with the K\"ahler form is a constant $B_0$ \cite{DK10}.

The zeros of a lowest Landau level state are called \emph{vortices}, as the phase of the wave function winds around them.

\paragraph{The magnetic field of the hyperbolic area form}
We use \eqref{eq-LLL} on the unit cell $\Sigma=\mathbb{H}/G$ of a hyperbolic crystal.
The canonical bundle $K_\Sigma$, with the Hermitian metric induced by the hyperbolic metric and the Chern connection, is the line bundle of a uniform magnetic field.
Indeed, in a local coordinate $z=x+iy$ in which the metric is $\rho\,|dz|^2$, the induced norm of $dz$ is $\|dz\|^2=2\rho^{-1}$, so the curvature of the Chern connection of $K_\Sigma$ is
\[ \bar\partial\partial\log\rho^{-1}=\partial_z\partial_{\bar z}\log\rho\;dz\wedge d\bar z=i\kappa\rho\,dx\wedge dy, \qquad \kappa=-\frac{2}{\rho}\,\partial_z\partial_{\bar z}\log\rho, \]
where $\kappa$ is the Gaussian curvature of the metric.
For the hyperbolic metric, $\kappa=-1$, e.g.\ $\rho=(\mathrm{Im}\,z)^{-2}$ on $\mathbb{H}$, so the curvature of the Chern connection is $-i\mathbf{B}$ with
\[ \mathbf{B}=\mathrm{vol}, \qquad \mathrm{vol}=\rho\,dx\wedge dy=\tfrac i2\rho\,dz\wedge d\bar z, \]
the hyperbolic area form.
Hence $L_{\mathbf{B}}=K_\Sigma$ is the line bundle of a uniform magnetic field of strength $B_0=1$.
The number of flux quanta is
\[ N_\Phi=\frac{1}{2\pi}\int_\Sigma\mathrm{vol}=2g-2=\deg K_\Sigma, \]
the area being $4\pi(g-1)$ by the Gauss--Bonnet theorem.

It follows from \eqref{eq-LLL} that the lowest Landau level of momentum $\mathcal{L}$ is $H^0(\Sigma,\mathcal{L}\otimes K_\Sigma)$.
If $\mathcal{L}\neq\mathcal{O}_\Sigma$, then $h^1(\mathcal{L}\otimes K_\Sigma)=h^0(\mathcal{L}^{-1})=0$, so the lowest Landau level has dimension $g-1$ by the Riemann--Roch theorem.
At the zero momentum $\Gamma$, the lowest Landau level jumps to $H^0(\Sigma,K_\Sigma)$, of dimension $g$.
This jump was found by Iengo and Li: in the magnetic field $L_{\mathbf{B}}=K_\Sigma$, the lowest Landau level has dimension $g$ for untwisted and $g-1$ for twisted boundary conditions \cite{IL94}.
The twist corresponds to magnetic flux through the handles of the unit cell, i.e.\ to a nonzero momentum.

The flux of the magnetic field $L_{\mathbf{B}}=K_\Sigma$ is the borderline one: for $N_\Phi>2g-2$ the lowest Landau level has the same dimension $N_\Phi+1-g$ at every momentum, as then $H^1(\Sigma,L)=0$ \cite{IL94}.
Studies of the Hall conductance as an adiabatic curvature on closed Riemann surfaces assume a flux above $2g-2$, so that the lowest Landau levels form a vector bundle over the Brillouin zone \cite{ASZ94, KMMW17}.

For the Hamiltonian without potential, the rest of the spectrum at momentum $\chi$ in the magnetic field $L_{\mathbf{B}}=K_\Sigma$ is the nonzero spectrum at momentum $\chi$ and zero magnetic field, shifted by $1$.
At the zero momentum, the shift is shown in \cite{IL94}, cf.\ \cite{CL25}.
The argument of \cite{IL94} applies at every momentum.
Namely, on sections of $K_\Sigma\otimes L_\chi$, i.e.\ $(1,0)$-forms with values in $L_\chi$, the Weitzenb\"ock formula and the K\"ahler identities give $\nabla^\ast\nabla=2\partial_\chi\partial_\chi^\ast+1$, where $\partial_\chi=\nabla_\chi^{1,0}$, while at zero magnetic field $\nabla_\chi^\ast\nabla_\chi=2\partial_\chi^\ast\partial_\chi$.
Therefore $\partial_\chi$ maps the eigenfunctions of energy $\mu>0$ at zero magnetic field onto the eigenfunctions of energy $1+\mu$ in the magnetic field $L_{\mathbf{B}}=K_\Sigma$.
The eigenfunctions of energy $1+\mu$, for all $\mu>0$, span the orthogonal complement of the lowest Landau level.
Hence the jump at $\Gamma$ is a band touching: the next band in the magnetic field $L_{\mathbf{B}}=K_\Sigma$ is the lowest band at zero magnetic field, shifted by $1$, so the next band meets the lowest Landau level exactly at $\Gamma$, where the lowest band at zero magnetic field has energy $0$ (\Cref{fig-band-touching}).

\subsection{Higgs bundles as hyperbolic crystals and momenta}\label{sec-crystals}
Following \cite{KR22}, we now summarize the relation between Higgs bundles and hyperbolic band theory and describe the setting of our paper.

Throughout the paper, the base curve is $C=\mathbb{P}^1$.
We consider Higgs bundles $(E,\phi)$ of rank two on $C$ with trace-free Higgs fields $\phi\colon E\to E\otimes\mathcal{O}(3)$.
The corresponding \emph{spectral curves} are
\[ \widetilde{C}_b=\{\lambda^2+b=0\}\subset\mathrm{Tot}(\mathcal{O}(3)), \qquad b=\det\phi, \]
where $\pi\colon\mathrm{Tot}(\mathcal{O}(3))\to\mathbb{P}^1$ is the projection and $\lambda$ the tautological section of $\pi^\ast\mathcal{O}(3)$.
The determinant $b=\det\phi$ is a section of $\mathcal{O}(3)^{\otimes2}=\mathcal{O}(6)$ and 
hence has six zeros counted with multiplicity if $b\neq0$.
In the generic case these zeros are all simple, so $\widetilde{C}_b$ is the smooth double cover of $\mathbb{P}^1$ branched at these six points and has genus two.

Let $\mathcal{M}$ be the moduli space of stable Higgs bundles 
$(E, \phi)$ with $\det(E) = \mathcal{O}(-3)$.
The \emph{Hitchin map} is the morphism
\[ h: \mathcal{M} \longrightarrow B=H^0(\mathcal{O}(6)) \cong \mathbb{C}^7, \qquad \qquad (E,\phi) \longmapsto\det\phi. \]
The fiber $h^{-1}(b)$ over a section $b \in B$ with simple zeros is canonically identified with the Jacobian $\mathrm{Jac}(\widetilde{C}_b)$ via the spectral correspondence \cite{BNR89}:
given $(E, \phi)$ with $\det(\phi) = b$, 
we define the associated spectral line bundle 
\[ \mathcal{L}_{(E, \phi)} \coloneqq \mathrm{coker}\left( \pi^\ast(E \otimes \mathcal{O}(-3)) \overset{\pi^\ast(\phi) - \lambda}{\longrightarrow} \pi^\ast(E)  \right) \in \mathrm{Jac}(\widetilde{C}_b) \]
and recover $(E, \phi)$ by taking the direct image,
\[ (E, \phi) = \left( \pi_\ast \mathcal{L}_{(E,\phi)}, \pi_\ast \lambda \right). \]

\paragraph{Smoothness}
Since the rank and the degree of $E$ are coprime, 
$\mathcal{M}$ is a smooth quasi-projective variety.
The moduli space $\mathcal{M}$ has dimension $9$: the Hitchin base has dimension $7$, and the Hitchin fibers have dimension $2$.

\begin{remark}[Poisson structures]\label{rem-poisson-str}
As $\dim\mathcal{M}=9$ is odd, $\mathcal{M}$ carries no holomorphic symplectic form, unlike the moduli spaces of $K$-twisted Higgs bundles \cite{Hit87a}; $\mathcal{M}$ does, however, carry Poisson structures \cite{Bottacin95, Markman94}.
\end{remark}

\paragraph{Hitchin section and zero momenta}
It follows from Grothendieck's theorem and the stability of the Higgs bundles that the underlying bundles $E$ are of two splitting types.
As $h^0(E)=h^0(\mathcal{L})$, the splitting type is determined by $\mathcal{L}$:
\[
	E \cong
	\begin{cases}
		\mathcal{O}\oplus\mathcal{O}(-3) \cong \pi_\ast \mathcal{O}_{\widetilde{C}_b}
		&\text{ for } \mathcal{L} \cong \mathcal{O}_{\widetilde{C}_b} \\
		\mathcal{O}(-1)\oplus\mathcal{O}(-2) 
		&\text{ for } \mathcal{L} \ncong \mathcal{O}_{\widetilde{C}_b}
	\end{cases}.
\]
A Higgs bundle $(E, \phi) \in h^{-1}(b)$ has a maximally unstable underlying bundle if and only if the spectral line bundle is $\mathcal{O}_{\widetilde{C}_b} \equiv 0 \in \mathrm{Jac}(\widetilde{C}_b)$,
i.e.\ the zero momentum in band theory language, 
\[ \Gamma_b=\mathcal{O}_{\widetilde{C}_b} \equiv 0 
\in \mathrm{Jac}(\widetilde{C}_b) \cong h^{-1}(b). \]
The Hitchin section is
\[ \mathrm{Hit}: B \longrightarrow \mathcal{M}, \qquad \qquad b \longmapsto \mathcal{O}_{\widetilde{C}_b} \in \mathrm{Jac}(\widetilde{C}_b) \cong h^{-1}(b), \]
which is, in terms of Higgs bundles, 
\begin{equation} \label{Hitchin-section}
	\mathrm{Hit}(b) = \left(E_0, \phi_b \right) \coloneqq  
	\left( \mathcal{O}\oplus\mathcal{O}(-3),\begin{pmatrix} 0 & -b \\ 1 & 0 \end{pmatrix} \right).
\end{equation}

In band theory, the Hitchin section associates to a crystal with unit cell $\widetilde{C}_b$ 
the zero momentum in the Brillouin zone. 

\section{Baker--Akhiezer divisors and the \BB stratification}
\label{sec-BB}

\subsection{Baker--Akhiezer divisors}
\label{sec-BAdiv}
As in \Cref{sec-crystals}, $\pi\colon\widetilde{C}_b\to C$ is the projection to the base curve $C=\mathbb{P}^1$.
Let $B_s \subset B$ be the open dense subset of sextics $b$ with simple zeros, i.e.\ with smooth spectral curve $\widetilde{C}_b$.
Let $(E,\phi)\in h^{-1}(b)$ with $b \in B_s$, 
and let $i\colon L\to E$ be a nonzero morphism from a line bundle $L$ on $C$.
Given such a triple $(L\xrightarrow{i}E,\phi)$,
on $\widetilde{C}_b$ consider the composition  
\begin{equation} \label{eqn-BA-section}
	\widetilde{s}_i(\phi): \pi^\ast L\xrightarrow{\ \pi^\ast i\ }\pi^\ast E\longrightarrow\mathcal{L}_{(E,\phi)},	
\end{equation}
where the second map is the quotient map in the definition of $\mathcal{L}_{(E,\phi)}$.

\begin{definition}\label{def-BA}
We say that the composition $\widetilde{s}_i(\phi)$,
regarded as a section of $\mathcal{L}_{(E,\phi)}\otimes\pi^\ast L^{-1}$,
and its zero divisor are respectively
the \emph{Baker--Akhiezer section} and \emph{Baker--Akhiezer divisor} of the triple $(L \xrightarrow{i} E,\phi)$.
\end{definition}

We denote the Baker--Akhiezer divisor by $\widetilde{D}_i(\phi)$, 
or more frequently by $\widetilde{D}$ if there is no confusion regarding the triple $(L \xrightarrow{i} E,\phi)$.
The morphism $i: L \rightarrow E$ corresponds to the Baker--Akhiezer section $\widetilde{s}_i(\phi)$ by adjunction.
The terminology \emph{Baker--Akhiezer section} and \emph{Baker--Akhiezer divisor}
was introduced in \cite{DT23}, after the Baker--Akhiezer functions of integrable systems \cite{Kri77, Dub81}.

The norm 
\[ D \equiv D_i(\phi) \coloneqq \mathrm{Nm}(\widetilde{D})=\sum_k\pi(\widetilde{p}_k)\]
of $\widetilde{D}=\sum_k\widetilde{p}_k $ is the divisor of the composition
\[ s_i(\phi) : 
L \overset{i}{\rightarrow} E \overset{\phi}{\rightarrow} E \otimes \mathcal{O}(3) \rightarrow L^{-1} \otimes \det(E) \otimes \mathcal{O}(3) \cong L^{-1}  \]
regarded as a section of $L^{-2}$.
Concretely, if $i$ is an embedding, i.e.\ nowhere vanishing, and 
if $\phi$ takes the local form 
\[ \phi = \begin{pmatrix} \phi_0 & \phi_- \\ \phi_+ & -\phi_0\end{pmatrix} \]
in the local frames adapted to the line subbundle $i(L)$, 
then $\phi_+$ represents $s_i(\phi)$ locally. 

\begin{remark} \label{rem-morphism-vanishes}
For $p\in C$, the divisor $\widetilde{D}_i(\phi)$ has $\pi^\ast(p)$ as a subdivisor if and only if 
the morphism $L \xrightarrow{i} E$ vanishes at $p$.
\end{remark}

\begin{remark}\label{rem-BA-concrete}
Let $i: L \hookrightarrow E$ be an embedding, and 
\( \phi= \left( \begin{smallmatrix}\phi_0&\phi_-\\\phi_+&-\phi_0\end{smallmatrix} \right) \)
in local frames adapted to $i(L)$.
If $D = \mathrm{div}\,\phi_+=\sum_kp_k$, then $\widetilde{D}=\sum_k\widetilde{p}_k$, where $\widetilde{p}_k=(p_k,-\phi_0(p_k))$ \cite{DT23}.
\end{remark}

By construction, the Baker--Akhiezer divisor $\widetilde{D}$ of $(L \overset{i}{\rightarrow} E, \phi)$ satisfies 
\[ \mathcal{L}_{(E, \phi)} \cong \pi^\ast(L)(\widetilde{D}). \]
Since
$ K_{\widetilde{C}_b} \cong \pi^\ast(\mathcal{O}(1))$,
every $\mathcal{L} \in \mathrm{Jac}(\widetilde{C}_b)\setminus \{0\}$ satisfies 
\begin{equation} \label{eq-LD}
\mathcal{L} \cong K_{\widetilde{C}_b}^{-1}(\widetilde{D})
\end{equation}
where $\widetilde{D}$ is the Baker--Akhiezer divisor of degree $2$ associated to the triple
\[ \left( \mathcal{O}(-1) \hookrightarrow \pi_\ast \mathcal{L}, \pi_\ast \lambda \right), \qquad \qquad \pi_\ast \mathcal{L} \cong \mathcal{O}(-1) \oplus \mathcal{O}(-2). \]

\paragraph{Linear systems of Baker--Akhiezer divisors}\label{sec-linsys}
Let $b\in B_s$, 
and $(E, \phi)\in h^{-1}(b)$ with spectral line bundle $\mathcal{L}$.
For $k \leq 0$, the Baker--Akhiezer sections 
\[ \widetilde{s}_{i}(\phi): \pi^\ast(\mathcal{O}(k)) \longrightarrow \pi^\ast(E) \longrightarrow \mathcal{L} \]
correspond by adjunction to the morphisms $\mathcal{O}(k) \xrightarrow{i} E$,
\[ \mathrm{Hom}(\mathcal{O}(k),E)=H^0(\mathbb{P}^1,E(-k))\cong H^0(\widetilde{C}_b,\mathcal{L}\otimes K_{\widetilde{C}_b}^{-k}), \qquad i\longmapsto \widetilde{s}_i(\phi). \]
For $k \leq -1$, the linear system $|\mathcal{L} \otimes K_{\widetilde{C}_b}^{-k} |$ 
of Baker--Akhiezer divisors is nonempty and has dimension 
\begin{equation}\label{eq-dimlinsys}
	\dim |\mathcal{L} \otimes K_{\widetilde{C}_b}^{-k} |
	= 
	\begin{cases}
		-2k-2 &\text{for } k \leq -2 \\
		0 &\text{for } k = -1, \mathcal{L} \neq \mathcal{O}_{\widetilde{C}_b} \\
		1 &\text{for } k = -1, \mathcal{L} = \mathcal{O}_{\widetilde{C}_b}.
	\end{cases}
\end{equation}
For $k = 0$, $|\mathcal{L} \otimes K_{\widetilde{C}_b}^{-k} | = | \mathcal{L} |$ is empty except for $\mathcal{L} = \mathcal{O}_{\widetilde{C}_b}$.

Thus the only nontrivial Baker--Akhiezer divisors whose linear system is zero-dimensional are those in \eqref{eq-LD}, 
i.e.\ those associated to a triple $(L \xrightarrow{i} E, \phi)$ with 
$E = \mathcal{O}(-1) \oplus \mathcal{O}(-2)$ and $L = \mathcal{O}(-1)$.
From now on, we call such a divisor the Baker--Akhiezer divisor of the Higgs bundle $(E,\phi)$.

On the other hand, the Baker--Akhiezer divisors of the triples 
\[ \left( \mathcal{O} \hookrightarrow E_0, \phi_b \right),
\qquad \qquad \left( \mathcal{O}(-3) \hookrightarrow E_0, \phi_b \right), \]
where $(E_0, \phi_b) = \mathrm{Hit}(b)$,
are respectively the trivial divisor and ramification divisor on $\widetilde{C}_b$.

\subsection{\BB stratification}
\label{sec-BBdec}

\paragraph{$\mathbb{C}^\ast$-action}
The moduli space of Higgs bundles admits a natural $\mathbb{C}^\ast$-action via 
$ t.[E,\phi]= [E,t\phi] $.
As the Hitchin map is proper \cite{Nitsure} and the $\mathbb{C}^\ast$-action on the Hitchin base $B$ is given by 
$t. b = t^2 b$, we have the following facts.

\begin{proposition}\label{lem-semiprojective}
\begin{enumerate}[label=(\roman*)]
\item For every $x\in\mathcal{M}$ the limit $\underset{t\to0}{\lim}t.x$ exists in $\mathcal{M}$, and it is a fixed point of the $\mathbb{C}^\ast$-action.
\item The fixed point locus $\mathcal{M}^{\mathbb{C}^\ast}$ is a closed subvariety of the nilpotent fiber $\mathcal{N}=h^{-1}(0)$, and $\mathcal{N}$ is a projective variety.
\item For $x\in\mathcal{M}$, the limit $\underset{t\to\infty}{\lim}t.x$ exists in $\mathcal{M}$ if and only if $x\in\mathcal{N}$.
\end{enumerate}
\end{proposition}

It follows from (i) and (ii) that $\mathcal{M}$ is \emph{semiprojective} in the sense of \cite{HRV15, HH21}: a quasi-projective variety with a $\mathbb{C}^\ast$-action is semiprojective if its fixed point locus is proper and $\underset{t\to0}{\lim}t.x$ exists for every $x$.

\paragraph{Flows associated to Hodge bundles}
We say that a fixed point $\alpha \in \mathcal{M}^{\mathbb{C}^\ast}$ is a \emph{Hodge bundle}.
Let $X$ be a connected component of $\mathcal{M}^{\mathbb{C}^\ast}$.
The component $X$ is smooth and projective.
At $\alpha \in X$, the tangent space decomposes into the subspaces of positive, zero and negative weights,
\[ T_\alpha\mathcal{M}=T^+_\alpha\oplus T^0_\alpha\oplus T^-_\alpha, \qquad T^0_\alpha=T_\alpha X. \]
We define the \emph{upward} and \emph{downward flows} of $\alpha$ to be, respectively,  
\[ \alpha^+=\{x\in\mathcal{M}\mid \underset{t\to0}{\lim}t.x=\alpha\}, \qquad \alpha^-=\{x\in\mathcal{M}\mid \underset{t\to\infty}{\lim}t.x=\alpha\},   \]
and the \emph{attractor} and \emph{repeller} of $X$ to be, respectively,
\[ X^\pm=\bigcup_{\alpha\in X}\alpha^\pm. \]

\begin{theorem}[\BB \cite{BB73}, {\cite[Prop.~2.1, Rem.~2.8]{HH21}}]\label{thm-BB}
Let $X$ be a connected component of $\mathcal{M}^{\mathbb{C}^\ast}$, let $\alpha\in X$, and let $n^\pm=\dim T^\pm_\alpha$.
\begin{enumerate}[label=(\roman*)]
\item The flows $\alpha^\pm$ are $\mathbb{C}^\ast$-invariant locally closed subvarieties of $\mathcal{M}$, $\mathbb{C}^\ast$-equivariantly isomorphic to $T^\pm_\alpha$. In particular, $\alpha^+$ is an affine space $\mathbb{C}^{n^+}$ on which $\mathbb{C}^\ast$ acts linearly with positive weights.
\item The limit map $X^+\to X$ is a $\mathbb{C}^\ast$-equivariant locally trivial fibration whose fibers are affine spaces of dimension $n^+$.
\item $\mathcal{M}=\bigsqcup_XX^+$ is the \emph{\BB decomposition}, and $\mathcal{N}=\bigsqcup_XX^-$ with its reduced structure, i.e.\ $\mathcal{N}$ is the \emph{core} of $\mathcal{M}$ \cite[Prop.~3.9]{HH21}.
\end{enumerate}
\end{theorem}

The proof of parts (i) and (ii) can be found in \cite{HH21} for normal semiprojective varieties, 
while part (iii) follows from \Cref{lem-semiprojective}.

\paragraph{\BB stratification of $\mathcal{M}$}
The following result classifies Hodge bundles in our setting.
The idea of the proof is that a \emph{stable} \emph{nilpotent} Hodge bundle 
$(E, \phi_n)$ with $E = \mathcal{O}(k) \oplus \mathcal{O}(-k-3)$ and 
$\mathcal{O}(k)$ destabilizing, i.e.\ $2k > - 3$, 
has $\ker(\phi_n) = \mathcal{O}(-k-3)$,
as otherwise $(E, \phi_n)$ is unstable. 
We refer to \cite{HH21} for the classification of $K$-twisted Hodge bundles of higher rank on curves of genus $g\geq2$. 

\begin{proposition}\label{prop-fixed}
The fixed point locus has two components $\mathcal{M}^{\mathbb{C}^\ast}=X_1\sqcup X_0$ 
where 
\[ X_1 \cong C^{(2)}, \qquad \qquad X_0 = \{ \ast \}. \]
Here $X_1 \cong C^{(2)}$ parametrizes Higgs bundles of the form
\[ \left( E_1, \phi_D \right) \coloneqq \left(\mathcal{O}(-1)\oplus\mathcal{O}(-2),\begin{pmatrix}0&0 \\ s_D&0\end{pmatrix}\right), \qquad D \in C^{(2)}, \]
where $s_D = s_{\mathcal{O}(-1) \hookrightarrow E_1}(\phi_D) \in H^0(\mathcal{O}(2))$ vanishes at $D$, 
and $X_0$ contains the single \emph{uniformizing} Higgs bundle
\[ \left(E_0, \phi_0 \right) = \left(\mathcal{O}\oplus\mathcal{O}(-3),\begin{pmatrix}0&0\\1&0\end{pmatrix}\right). \]
\end{proposition}

The $\mathbb{C}^\ast$-action preserves the underlying bundle $E$ and the divisor of $s_i(\phi)$ for $L \overset{i}{\hookrightarrow} E$ the destabilizing line subbundle, since $s_i(t\phi)=t\,s_i(\phi)$.
In our setting, the limit of $(E,\phi)$ as $t\to0$ is the Hodge bundle with the same underlying bundle and the same divisor.
We refer to \cite[Prop.~3.4]{HH21} for the proof.
\begin{proposition}\label{prop-flows}
(i) Let $\alpha = (E_1, \phi_D) \in X_1$. 
Then a Higgs bundle $(E, \phi) \in \mathcal{M}$ is contained in $\alpha^+$ 
if and only if 
\[ \left( E, \phi \right) \cong \left(\mathcal{O}(-1)\oplus\mathcal{O}(-2),\begin{pmatrix} \ast & \ast \\ s_D & \ast \end{pmatrix}\right), \qquad D \in C^{(2)}. \]
Consequently, the underlying bundles of Higgs bundles $(E, \phi) \in X_1^+$ are $\mathcal{O}(-1) \oplus \mathcal{O}(-2)$. 
\\
(ii) For $\alpha$ the uniformizing Higgs bundle $(E_0, \phi_0)$,
the upward flow $\alpha^+$ is the Hitchin section \eqref{Hitchin-section}.
\end{proposition}

Identifying $X_1$ with $C^{(2)}$, we write $D \in C^{(2)}$ for Hodge bundles in $X_1$ from now on. 
The \BB stratification of $\mathcal{M}$ is the decomposition of $\mathcal{M}$ into attractors of Hodge bundle components,
\[ \mathcal{M}= X_1^+\sqcup X_0^+ = X_1^+\sqcup \mathrm{Hit}(B),  \]
where $X_1^+$ is the open dense stratum, i.e.\ $\overline{X_1^+} = \mathcal{M}$.
For all Hodge bundles $\alpha$ we have $\dim \alpha^+ = \dim B = 7$.

\paragraph{The nilpotent fiber}
The nilpotent fiber $\mathcal{N}$ is the disjoint union of the repellers of the Hodge bundle components $X_1$ and $X_0$ (\Cref{thm-BB}(iii)),  
\[ \mathcal{N} = X_1^- \sqcup X_0^-, \]
where $\dim X_1^- = \dim X_0^- = \dim \mathcal{N} = 2$.
The following proposition describes these repellers explicitly.

\begin{proposition} \label{prop-nilpotent-component}
The repeller $X_1^-$ is precisely $X_1$.
The repeller $X_0^-$ consists of the uniformizing Higgs bundle and of the nilpotent Higgs bundles of the form
\[ (E, \phi_{u, v}) = \left( \mathcal{O}(-1) \oplus \mathcal{O}(-2),
\begin{pmatrix}
    uv & -u^2 \\ v^2 & -uv
\end{pmatrix}
\right), \]
where $u\in H^0(\mathcal{O}(2))$ and $v\in H^0(\mathcal{O}(1))$ have no common zero.
The kernel of $\phi_{u,v}$ is the line subbundle
\[ (u,v)^T: \mathcal{O}(-3) \hookrightarrow \mathcal{O}(-1) \oplus \mathcal{O}(-2). \]
Finally, the universal extension of $\mathcal{O}$ by $\mathcal{O}(-3)$ gives a $\mathbb{C}^\ast$-equivariant isomorphism of varieties $H^1(\mathcal{O}(-3))\cong X_0^-$, where $t\in\mathbb{C}^\ast$ acts on $H^1(\mathcal{O}(-3))\cong\mathbb{C}^2$ by $e\mapsto t^{-1}e$.
This isomorphism maps the origin to the uniformizing Higgs bundle.
\end{proposition}

\begin{proof}
Let $(E,\phi)\in\mathcal{N}$.
Then $\phi^2=0$ by the Cayley--Hamilton theorem, and $\phi\neq0$ as in the argument before \Cref{prop-fixed}, hence $\ker\phi=\mathcal{O}(c)$ is the only $\phi$-invariant line subbundle of $E$.
Moreover, $\phi$ induces a nonzero map
\[ E/\ker\phi\cong\mathcal{O}(-3-c)\longrightarrow\ker\phi\otimes\mathcal{O}(3)=\mathcal{O}(c+3). \]
Hence stability gives $c\leq-2$, and the nonzero map gives $c\geq-3$.

If $c=-2$, then $E=\mathcal{O}(-1)\oplus\mathcal{O}(-2)$, since $H^1(\mathcal{O}(-1))=0$.
Hence $(E,\phi)$ is a Hodge bundle in $X_1$.
If $c=-3$, then $E=E_e$ is the extension of $\mathcal{O}$ by $\mathcal{O}(-3)$ with class $e\in H^1(\mathcal{O}(-3))\cong\mathbb{C}^2$, and $\phi$ is the composition
\[ E_e\to\mathcal{O}\xrightarrow{\ \kappa\ }\mathcal{O}(-3)\otimes\mathcal{O}(3)\to E_e\otimes\mathcal{O}(3) \]
for a constant $\kappa\neq0$.
Rescaling the subbundle $\mathcal{O}(-3)$ by $\nu\in\mathbb{C}^\ast$ gives an isomorphism of $(E_e,\kappa)$ with $(E_{\nu e},\nu\kappa)$.
Hence every such Higgs bundle is isomorphic to $(E_e,1)$ for some $e$.
The class $e$ is unique, since an isomorphism preserves the kernels.
So the universal extension gives an injective morphism $H^1(\mathcal{O}(-3))\to\mathcal{M}$, $e\mapsto(E_e,1)$, onto the locus $c=-3$.

Since $t.(E_e,1)\cong(E_{e/t},1)$ limits to the uniformizing Higgs bundle $(E_0,1)$ as $t\to\infty$, the Higgs bundles with $c=-3$ form $X_0^-$, and those with $c=-2$ form $X_1^-=X_1$.
Hence $e\mapsto(E_e,1)$ is a bijective morphism $H^1(\mathcal{O}(-3))\to X_0^-$, which is $\mathbb{C}^\ast$-equivariant for the action $t.e=t^{-1}e$.
Since $X_0^-$ is smooth by \Cref{thm-BB}(i), this morphism is an isomorphism of varieties by Zariski's main theorem.

Finally, for $e\neq0$ the connecting map $H^0(\mathcal{O})\to H^1(\mathcal{O}(-3))$ maps $1$ to $e$, hence $H^0(E_e)=0$ and $E_e\cong\mathcal{O}(-1)\oplus\mathcal{O}(-2)$ by Grothendieck's theorem.
Thus the inclusion of $\ker\phi$ is a column $(u,v)^T$ without common zero, so the quotient map is the row $(v,-u)$ up to a scalar.
Hence $\phi$ is a nonzero multiple of $(u,v)^T(v,-u)$, i.e.\ $\phi=\phi_{u,v}$ after rescaling $(u,v)$.
Conversely, $\phi_{u,v}^2=0$ and its kernel $(u,v)^T\colon\mathcal{O}(-3)\hookrightarrow E$ is the only $\phi_{u,v}$-invariant line subbundle of $E$, hence $(E,\phi_{u,v})$ is stable with $c=-3$.
\end{proof}

\section{Degeneration of Baker--Akhiezer divisors} 
\label{sec-degen}

\subsection{Wobbly Hodge bundles}
\label{sec-wobbly}

\paragraph{Very stable and wobbly fixed points}
Very stable bundles were introduced by Drinfeld in the context of the geometric Langlands program \cite{Laumon88}.
Donagi--Pantev called the stable bundles that are not very stable \emph{wobbly} \cite{DP09}.
These notions were later generalized by Hausel--Hitchin \cite{HH21} 
to the setting of fixed points in semiprojective varieties,
whose \emph{core} plays the role of the nilpotent fiber.

\begin{definition} [{\cite{HH21}}]
\label{def-verystable}
A fixed point $\alpha \in \mathcal{M}^{\mathbb{C}^\ast}$ is called \emph{very stable} 
if the only point of its upward flow which lies in the nilpotent fiber is $\alpha$ itself, i.e.\ 
\[ \alpha^+\cap\mathcal{N}=\{\alpha\}. \]
Otherwise, $\alpha$ is called \emph{wobbly}.
\end{definition}

For $X \in \pi_0(\mathcal{M}^{\mathbb{C}^\ast})$ containing $\alpha$,
the condition that $\alpha$ is wobbly is equivalent to the existence of 
another component $Y \in \pi_0(\mathcal{M}^{\mathbb{C}^\ast})$ such that 
\[ \alpha \in X \cap \overline{Y^-}. \]
The point $\alpha$ lies in $\overline{Y^-}$ if and only if there exists $x \in \alpha^+$ such that 
$\underset{t \rightarrow \infty}{\lim} t.x \in Y$.

In our setting, since $X_0^+$ is the Hitchin section itself, the uniformizing Higgs bundle $(E_0, \phi_0)$ is very stable. 
The wobbly locus is therefore
\[ W = X_1 \cap \overline{X_0^-}. \]

Before determining the wobbly locus explicitly, we recall an equivalent characterization of very stable points. 
We refer to \cite{PPe19, HH21} for the proofs.

\begin{proposition}\label{prop-verystable}
A fixed point $\alpha\in\mathcal{M}^{\mathbb{C}^\ast}$ is very stable if and only if its upward flow $\alpha^+$ is closed in $\mathcal{M}$.
In this case the restriction $h|_{\alpha^+}\colon \alpha^+\to B$ of the Hitchin map is proper, hence finite and flat.
Furthermore, $\alpha^+$ is a multisection of the Hitchin fibration, 
i.e.\ every Hitchin fiber contains the same nonzero number of points in $\alpha^+$, counted with multiplicity.
\end{proposition}

\paragraph{The wobbly locus}
The following characterization of the wobbly locus is analogous to that for $K$-twisted Higgs bundles \cite[Thm.~1.2]{HH21}.
Let $\Delta = \{2p \mid p \in C\}\subset X_1\cong C^{(2)}$ be the diagonal.

\begin{corollary}\label{cor-wobbly}
The wobbly locus is precisely the diagonal,
\[ W = \Delta. \]
Equivalently, the Hodge bundle $(E_1,\phi_D)$ is wobbly if and only if $D = 2p$ for some $p \in C$.
\end{corollary}
\begin{proof}
Let $f\colon X_1^+\to X_1$, $x\mapsto\underset{t\to0}{\lim}t.x$, be the limit map of \Cref{thm-BB}(ii).
The map $f$ is continuous and is the identity on $X_1$.
Recall from \Cref{prop-nilpotent-component} that the points of $X_0^-$ other than $\mathrm{Hit}(0)$ are of the form $(E,\phi_{u,v})$.
It follows from \Cref{prop-flows}(i) that $f(E,\phi_{u,v}) = 2\,\mathrm{div}(v)$, which lies in $\Delta\cap\overline{X_0^-}$ as $X_0^-$ is $\mathbb{C}^\ast$-invariant.
Since we can freely choose the zero $\mathrm{div}(v)$, we have $\Delta \subset W = X_1 \cap \overline{X_0^-}$.
Conversely, to show that $W \subset \Delta$, 
observe that $f$ maps $X_0^-\smallsetminus\{\mathrm{Hit}(0)\}$ into $\Delta$, which is closed in $X_1$.
By continuity of $f$, a wobbly point $\alpha \in W  = X_1 \cap \overline{X_0^-\smallsetminus\{\mathrm{Hit}(0)\}}$ satisfies
\[ \alpha=f(\alpha)\in\overline{f\big(X_0^-\smallsetminus\{\mathrm{Hit}(0)\}\big)}\subseteq\Delta. 
\qedhere \]
\end{proof}

For $\alpha\in X_1$ very stable, it follows from \Cref{prop-verystable} that $h|_{\alpha^+}\colon\alpha^+\to B$ is finite, flat and surjective, of degree $4$ by \Cref{lem-BA}(i).
On the other hand, if $\alpha \in W$, 
the fiber of $h|_{\alpha^+}$ over $0$ contains a $\mathbb{C}^\ast$-orbit of nilpotent Higgs bundles $(E,\phi_{u,v})$, so that $h|_{\alpha^+}$ is not finite.

\paragraph{Intersection of upward flows with Hitchin fibers}
The next result describes the intersection of upward flows $\alpha^+$ for $\alpha \in X_1$ with smooth Hitchin fibers.
An analogue for $K$-twisted Higgs bundles of rank $n$ can be found in \cite[Prop.~5.18]{HH21}.

\begin{proposition}\label{lem-BA}
Let $b\in B_s$.
Let $\alpha\equiv D=p_1+p_2\in X_1$, and let $D_b\leq D$ be the subdivisor consisting of branch points of $\pi: \widetilde{C}_b \rightarrow C$.
\begin{enumerate}[label=(\roman*)]
\item Suppose $\alpha$ is very stable, i.e.\ $p_1\neq p_2$.
Then $\alpha^+\cap h^{-1}(b)$ consists of exactly $2^{2-\deg D_b}$ Higgs bundles with spectral line bundles $K_{\widetilde{C}_b}^{-1}(\widetilde{D})$ for the divisors $\widetilde{D}$ with $\mathrm{Nm}(\widetilde{D})=D$.
In particular, it consists of $4$ Higgs bundles if neither $p_1$ nor $p_2$ is a branch point.
\item Suppose $\alpha$ is wobbly, i.e.\ $p_1=p_2=p$.
Then $\alpha^+\cap h^{-1}(b)$ is empty if $p$ is a branch point of $\pi$, and otherwise consists of the two Higgs bundles with spectral line bundles $K_{\widetilde{C}_b}^{-1}(2\widetilde{p})$, $\pi(\widetilde{p})=p$.
\end{enumerate}
\end{proposition}

\begin{proof}
We prove the statement on the emptiness of $\alpha^+\cap h^{-1}(b)$, following \cite{HH21}.
The rest of the proposition concerning $\pi_\ast \mathcal{L} \in \alpha^+\cap h^{-1}(b)$ follows by considering Baker--Akhiezer divisors of triples 
$(\mathcal{O}(-1) \hookrightarrow \pi_\ast \mathcal{L}, \pi_\ast \lambda)$.
If $p$ is a branch point, then 
$2\widetilde{p} = \pi^\ast(p)$ where 
$\widetilde{p}$ is the ramification point of $\widetilde{C}_b$ over $p$.
If the Baker--Akhiezer divisor of the triple $(\mathcal{O}(-1) \hookrightarrow \pi_\ast \mathcal{L}, \pi_\ast \lambda)$ is 
$2\widetilde{p}$, then it follows from \Cref{rem-morphism-vanishes} that
the morphism $\mathcal{O}(-1) \rightarrow  \pi_\ast \mathcal{L}$ vanishes at $p$. 
So there is an embedding from $\mathcal{O} \cong \mathcal{O}(-1) \otimes \mathcal{O}(p)$ to $\pi_\ast \mathcal{L}$,
and hence $(\pi_\ast \mathcal{L}, \pi_\ast \lambda)$ is the intersection of the Hitchin section with $h^{-1}(b)$.
\end{proof}

\subsection{Crossing to the Hitchin section}
\label{sec-crossing}
Let $b\in B_s$.
Recall from \Cref{sec-BAdiv} that the Baker--Akhiezer divisor $\widetilde{D}$ of $(E,\phi)\in X_1^+\cap h^{-1}(b)$ is the Baker--Akhiezer divisor of the triple $(\mathcal{O}(-1)\hookrightarrow E,\phi)$, whose inclusion is unique up to scaling.
The spectral line bundle 
of $(E,\phi)$ is the image of $\widetilde{D}$ under the morphism
\begin{equation}\label{eq-blowup}
\widetilde{C}_b^{(2)}\longrightarrow\mathrm{Jac}(\widetilde{C}_b)=h^{-1}(b), 
\qquad \widetilde{D} \longmapsto \mathcal{O}_{\widetilde{C}_b}(\widetilde{D}) \otimes K_{\widetilde{C}_b}^{-1},
\end{equation}
where $\widetilde{D}$ is in the complement of the linear system
\[ |K_{\widetilde{C}_b}| = |\pi^\ast \mathcal{O}(1) | =  \{\pi^\ast p \mid p \in \mathbb{P}^1\} \]
in $\widetilde{C}_b^{(2)}$.
The image of the curve $|K_{\widetilde{C}_b}|$ is the single point 
\[ \mathcal{O}_{\widetilde{C}_b}(\pi^\ast p) \otimes K_{\widetilde{C}_b}^{-1} = \pi^\ast( \mathcal{O}(1)) \otimes K_{\widetilde{C}_b}^{-1} \cong \mathcal{O}_{\widetilde{C}_b} \]
in the Hitchin section.

The following proposition describes the limits of the Baker--Akhiezer divisors of the Higgs bundles in a smooth Hitchin fiber when the two points of $D = \mathrm{Nm}(\widetilde{D})$ collide.

\begin{proposition}\label{prop-crossing}
Let $b\in B_s$, and let $x_u\in X_1^+\cap h^{-1}(b)$, $0<u\le1$, be a continuous family whose limits $\underset{t\to0}{\lim}t.x_u$ are the Hodge bundles with divisors $p_1(u)+p_2(u)$, where $p_1(u),p_2(u)\to p$ as $u\to0$.
Then as $u \to 0$, the limit of the Baker--Akhiezer divisors $\widetilde{D}_u$ of $x_u$ is either 
\begin{enumerate}[label=(\roman*)]
\item a double point, 
$\underset{u \rightarrow 0}{\lim} \widetilde{D}_u = 2\widetilde{p}$, where $\mathrm{Nm}(\widetilde{p}) = p$ is not a branch point. 
In this case, $x_0 = \underset{u \rightarrow 0}{\lim} x_u$ 
is contained in $\alpha^+ \cap h^{-1}(b)$ where $\alpha = 2p \in \Delta$ is wobbly;
\item $\underset{u \rightarrow 0}{\lim} \widetilde{D}_u = \pi^\ast p$.
In this case, $x_u$ limits to the point $\mathrm{Hit}(b)$ of the Hitchin section.
\end{enumerate}
\end{proposition}

\begin{remark}\label{rem-collisions}
The two cases of \Cref{prop-crossing} also occur for rank-two Higgs bundles on compact Riemann surfaces of genus $g \ge 2$ \cite{DFT26}.
In that setting, the families of Higgs bundles in the two cases are analogues of families of flat connections 
in which two apparent singularities of opers collide.
Under specific conditions, the families of flat connections limit to smaller strata of Simpson's partial oper stratification of the de Rham moduli space.

Collisions of spectral zeros also appear in \cite{RS26}, which constructs quantum error-correcting codes from twisted cyclic and argyle quiver moduli \cite{RS18, RS21} on $\mathbb{P}^1$.
\end{remark}

\begin{remark}\label{rem-direction}
In \Cref{prop-crossing}(ii), the point $p$ records the direction in which $x_u$ approaches $\mathrm{Hit}(b)$.
Indeed, in genus two the morphism \eqref{eq-blowup} is the blow-up of $h^{-1}(b)$ at $\mathrm{Hit}(b)$.
The exceptional curve $|K_{\widetilde{C}_b}|$ \cite{BL04} is the projectivized tangent plane $\mathbb{P}(T_{\mathrm{Hit}(b)}h^{-1}(b))=\mathbb{P}(H^0(K_{\widetilde{C}_b})^\vee)$, which is the base $C=\mathbb{P}^1$, as the canonical map of $\widetilde{C}_b$ is $\pi$.
Under this identification, the point $\pi^\ast p$ of the exceptional curve is the direction $p$.
Conversely, if $x_u$ approaches $\mathrm{Hit}(b)$ in the direction $p$, the Hodge bundles $\underset{t\to0}{\lim}t.x_u$ limit to the wobbly Hodge bundle with divisor $2p$.
Since the limiting wobbly Hodge bundle depends on the direction $p$, the limit map $X_1^+\to X_1$ of \Cref{thm-BB}(ii) has no continuous extension to $\mathrm{Hit}(b)$.
\end{remark}

\section{Baker--Akhiezer sections and Landau levels}
\label{sec-BA}

In this section, we interpret in band theory the Baker--Akhiezer sections $\widetilde{s}_i(\phi)$ of the points of the upward flows $D^+$, $D\in X_1\cong C^{(2)}$ (\Cref{lem-BA}).

\subsection{Lowest Landau levels}
\label{sec-LLL}
Let $b\in B_s$, and let $\mathcal{L}$ be the momentum of $x=(E,\phi)\in h^{-1}(b)$.
We place the crystal with unit cell $\widetilde{C}_b$ in the uniform magnetic field $\mathbf{B}=\mathrm{vol}$.
Recall from \Cref{sec-fields} that the line bundle of $\mathbf{B}$ is
\[ L_{\mathbf{B}} = K_{\widetilde{C}_b}\cong\pi^\ast\mathcal{O}(1), \qquad N_\Phi=\deg K_{\widetilde{C}_b}=2g-2=2. \]
The wave functions of momentum $\mathcal{L}$ are then the sections of $\mathcal{L}\otimes K_{\widetilde{C}_b}$.
It follows from \eqref{eq-LLL} that the lowest Landau level of momentum $\mathcal{L}$ is
\[ H^0(\mathcal{L}\otimes K_{\widetilde{C}_b}). \]
As $\mathcal{L}$ varies, these are the lowest Landau levels computed by Iengo and Li on any closed hyperbolic surface \cite{IL94}. 
In our setting, it follows from \Cref{sec-fields} that the lowest Landau level $H^0(\mathcal{L}\otimes K_{\widetilde{C}_b})$ is one-dimensional at every nonzero momentum and two-dimensional at the zero momentum $\Gamma_b$, where it is $H^0(K_{\widetilde{C}_b})$ (cf.\ \eqref{eq-dimlinsys}).

In the Higgs bundle framework, placing the crystal in the magnetic field $\mathbf{B} = \mathrm{vol}$ corresponds to tensoring the spectral line bundle with $K_{\widetilde{C}_b}$.\footnote{In \cite{KR22}, Kienzle and the second author proposed to treat a magnetic field by replacing the line bundles of degree zero with line bundles of nonzero degree.}

At $\Gamma_b=\mathrm{Hit}(b)$, the inclusion $i$ of the summand $\mathcal{O}$ of $E=\mathcal{O}\oplus\mathcal{O}(-3)$ corresponds by adjunction to the constant function $\widetilde{s}_i(\phi)=1$.
The Baker--Akhiezer divisor of the triple $(\mathcal{O}\xrightarrow{i}E,\phi)$ is the trivial divisor.
In band theory, $\widetilde{s}_i(\phi)=1$ is the ground state at zero magnetic field, of energy $0$, and has no vortices.
\Cref{tab-fields} lists the lowest Landau levels and their states in the magnetic fields $L_{\mathbf{B}}=K_{\widetilde{C}_b}^n$, $n\geq0$.

\begin{table}[ht]
\centering
\small
\setlength{\tabcolsep}{3.5pt}
\renewcommand{\arraystretch}{1.5}
\begin{tabular}{l|l|l|l|l}
\begin{tabular}[t]{@{}l@{}}\textbf{Line}\\ \textbf{bundle} $L$\end{tabular} & \begin{tabular}[t]{@{}l@{}}\textbf{Magnetic}\\ \textbf{field} $L_{\mathbf{B}}$\end{tabular} & \begin{tabular}[t]{@{}l@{}}\textbf{Level at}\\ $\mathcal{L}\neq\mathcal{O}_{\widetilde{C}_b}$\end{tabular} & \textbf{Level at} $\Gamma_b$ & \textbf{Vortices} \\ \hline
$\mathcal{O}$ & $\mathcal{O}_{\widetilde{C}_b}$ (zero), $N_\Phi=0$ & $0$ & $\mathbb{C}\cdot1$ & none \\
$\mathcal{O}(-1)$ & $K_{\widetilde{C}_b}$, $N_\Phi=2$ & $\mathbb{C}\,\widetilde{s}_i(\phi)$ & $\pi^\ast H^0(\mathcal{O}(1))\cong\mathbb{C}^2$ & \begin{tabular}[t]{@{}l@{}}$\widetilde{D}$, fixed by $\mathcal{L}$;\\ $\pi^\ast p$ at $\Gamma_b$\end{tabular} \\
$\mathcal{O}(-n)$, $n\geq2$ & $K_{\widetilde{C}_b}^n$, $N_\Phi=2n$ & $\mathbb{C}^{2n-1}$ & $\mathbb{C}^{2n-1}$ & $2n$, with $2n-2$ free
\end{tabular}
\vspace{8pt} 
\caption{The lowest Landau levels $H^0(\mathcal{L}\otimes L_{\mathbf{B}})$ of momentum $\mathcal{L}$ in the uniform magnetic fields $\mathbf{B}=n\,\mathrm{vol}$, $L_{\mathbf{B}}=K_{\widetilde{C}_b}^n$ (for $n=0$, the states of energy $0$), and the vortices of their states. The states are the Baker--Akhiezer sections $\widetilde{s}_i(\phi)$ of the triples $(L\xrightarrow{i}E,\phi)$ over $b\in B_s$ with $L=\mathcal{O}(-n)$ (\Cref{sec-linsys}), so the vortices are the Baker--Akhiezer divisors $\widetilde{D}$. For $n=1$, $i$ is the inclusion of the summand $\mathcal{O}(-1)$ of $E$ away from $\Gamma_b$.
At $\Gamma_b$, the state $\pi^\ast s_p$, with $s_p\in H^0(\mathcal{O}(1))$ vanishing at $p\in\mathbb{P}^1$, has the vortices $\pi^\ast p$.}\label{tab-fields}
\end{table}

\paragraph{The jump at the zero momentum}
Recall from \Cref{sec-fields} that, in band theory, the jump at $\Gamma_b$ is a band touching (\Cref{fig-band-touching}, cf.\ \cite{IL94}).

In the Higgs bundle framework, this jump happens at the intersection of the Hitchin fiber $h^{-1}(b)$ with the Hitchin section.
Since $h^0(E(1))=h^0(\mathcal{L}\otimes K_{\widetilde{C}_b})$, the jump corresponds to the change of the splitting type of the underlying bundles 
from $\mathcal{O}(-1)\oplus\mathcal{O}(-2)$ to $\mathcal{O}\oplus\mathcal{O}(-3)$, i.e.\ from the stratum $X_1^+$ to $X_0^+$.

At every nonzero momentum, 
the two vortices of the lowest Landau level state form the Baker--Akhiezer divisor $\widetilde{D}$, which determines the momentum 
$\mathcal{L}\cong\mathcal{O}_{\widetilde{C}_b}(\widetilde{D})\otimes K_{\widetilde{C}_b}^{-1}$.
At $\Gamma_b$, every state is up to scaling the Baker--Akhiezer section $\pi^\ast s_p$ of a morphism $s_p\colon\mathcal{O}(-1)\to\mathcal{O}\subset E$,
which vanishes at a unique $p\in\mathbb{P}^1$.
So the vortices of $\pi^\ast s_p$ form $\pi^\ast p$, which is invariant under the inversion.
If a family of nonzero momenta limits to $\Gamma_b$ in the direction $p$, the vortices of its lowest Landau level states limit to $\pi^\ast p$ (\Cref{rem-direction}).
Up to scaling, each of these states is the canonical section of the divisor of its vortices.
As canonical sections vary continuously with the divisor in $\widetilde{C}_b^{(2)}$, the states limit, up to scaling, to the state $\pi^\ast s_p$.
So every nonzero state of the jumped level is such a limit.
The limiting states are labeled by the complex directions of approach.

\paragraph{Magnetic fields of higher flux}
For an integer $n\geq1$, the uniform magnetic field $\mathbf{B}=n\,\mathrm{vol}$ has the line bundle $L_{\mathbf{B}}=K_{\widetilde{C}_b}^n$ and $N_\Phi=2n$ flux quanta.
It follows from \eqref{eq-LLL} and \Cref{sec-linsys} that the lowest Landau level of momentum $\mathcal{L}$ in this magnetic field has energy $n$ and is
\[ H^0(\widetilde{C}_b,\mathcal{L}\otimes K_{\widetilde{C}_b}^n)\cong\mathrm{Hom}(\mathcal{O}(-n),E). \]
So the linear systems $|\mathcal{L}\otimes K_{\widetilde{C}_b}^n|$ of \Cref{sec-linsys} are the projectivized lowest Landau levels in the magnetic fields $\mathbf{B}=n\,\mathrm{vol}$.

Thus, in the Higgs bundle framework, the line bundle $L=\mathcal{O}(-n)$ of a triple, i.e.\ $k=-n$ in \Cref{sec-linsys}, determines the magnetic field: $L_{\mathbf{B}}=K_{\widetilde{C}_b}^n\cong\pi^\ast L^{-1}$.
The Baker--Akhiezer sections $\widetilde{s}_i(\phi)$ of the triples $(\mathcal{O}(-n)\xrightarrow{i}E,\phi)$ are the lowest Landau level states in the magnetic field $\mathbf{B}=n\,\mathrm{vol}$, so the Baker--Akhiezer divisors consist of the $2n$ vortices of these states.

In every magnetic field $\mathbf{B}=n\,\mathrm{vol}$, $n\geq1$, the vortices determine the momentum: $\mathcal{L}\cong\mathcal{O}_{\widetilde{C}_b}(\widetilde{D})\otimes K_{\widetilde{C}_b}^{-n}$.
It follows from this formula and Abel's theorem that the image of the vortex divisor $\widetilde{D}$ under the Abel--Jacobi map is fixed by the momentum.
On a torus, this is the fact that the sum of the vortices of a lowest Landau level state is fixed by the boundary conditions, i.e.\ by the momentum, cf.\ \cite{HR85}.

For $n\geq2$, the flux exceeds the borderline flux $2g-2=2$ of \Cref{sec-fields}, so the lowest Landau level has dimension $2n-1=N_\Phi+1-g$ at every momentum, $\Gamma_b$ included.
The projectivization of this degenerate level is the linear system $|\mathcal{L}\otimes K_{\widetilde{C}_b}^n|\cong\mathbb{P}^{2n-2}$ of Baker--Akhiezer divisors.
The vortices of a superposition of lowest Landau level states move in this linear system.
The $2n$ vortices are coordinates on the projectivized level: $2n-2=N_\Phi-g$ of them can be prescribed at generic points, and the remaining $g=2$ are then fixed by the momentum.
A weak potential, which splits the level by degenerate perturbation theory, selects particular vortex configurations.

\paragraph{Jumps and Brill--Noether loci}
It follows from \eqref{eq-LLL} that the momenta at which the lowest Landau level in a magnetic field of $d$ flux quanta has dimension at least $r+1$ form a translate of the Brill--Noether locus \cite{ACGH85}
\[ W^r_d(\widetilde{C}_b)=\{M\in\mathrm{Pic}^d(\widetilde{C}_b) : h^0(M)\geq r+1\}. \]
At the borderline flux $d=2g-2$, the Riemann--Roch theorem gives $W^{g-1}_{2g-2}(\widetilde{C}_b)=\{K_{\widetilde{C}_b}\}$.
In genus two, $K_{\widetilde{C}_b}\cong\pi^\ast\mathcal{O}(1)$ is the hyperelliptic line bundle.
The pencil $|K_{\widetilde{C}_b}|$ of the fibers of $\pi$ is the projectivized jumped level.

The Riemann--Roch theorem and Serre duality give $h^0(\mathcal{L}\otimes K_{\widetilde{C}_b})=g-1+h^0(\mathcal{L}^{-1})$ at every momentum $\mathcal{L}$.
So the dimension of the lowest Landau level in the magnetic field $\mathbf{B}=\mathrm{vol}$ exceeds its generic value $g-1$ exactly by the number of states of energy $0$ at zero magnetic field and the opposite momentum $\mathcal{L}^{-1}$ (the first two rows of \Cref{tab-fields}).
This number is $1$ at $\Gamma_b$, where the constant function is such a state, and $0$ at every other momentum.

\subsection{Theta functions and the Fourier--Mukai transform}
\label{sec-theta}
We now summarize the basics of the Fourier--Mukai transform before relating Baker--Akhiezer sections to theta functions.
In \Cref{sec-4d}, we then interpret the theta functions as the lowest Landau level states of a four-dimensional Euclidean crystal with the same Brillouin zone.

\paragraph{The Jacobian and its dual}
Fix $b\in B_s$, and let $J=h^{-1}(b)=\mathrm{Jac}(\widetilde{C}_b)$.
We recall from \cite{HH21} how $J$ is identified with its dual $\widehat J=\mathrm{Pic}^0(J)$ of translation-invariant line bundles, using standard facts on Jacobians of genus two curves \cite{BL04}.

Fix a Weierstrass point $w_0\in\widetilde{C}_b$, i.e.\ a ramification point of the hyperelliptic curve, 
so that $K_{\widetilde{C}_b}\cong\mathcal{O}_{\widetilde{C}_b}(2w_0)$.
The \emph{Abel--Jacobi map}
\[ \mathrm{A}\colon\widetilde{C}_b\to J, \qquad y\mapsto\mathcal{O}_{\widetilde{C}_b}(y-w_0), \]
maps $\widetilde{C}_b$ isomorphically onto the \emph{theta divisor} $\Theta=\mathrm{A}(\widetilde{C}_b)$.
The divisor $\Theta$ is symmetric, i.e.\ $-\Theta=\Theta$.
Its class $\theta\in H^2(J,\mathbb{Z})$ is the \emph{principal polarization} of $J$.

The divisor $\Theta$ defines the isomorphism
\[ \varphi_\Theta\colon J\to\widehat J, \qquad a\mapsto t_a^\ast\mathcal{O}_J(\Theta)\otimes\mathcal{O}_J(-\Theta), \]
where $t_a\colon J\to J$, $z\mapsto z+a$, 
is the translation by $a\in J$.
Pulling back line bundles along $\mathrm{A}$ 
defines an isomorphism $\mathrm{A}^\ast\colon\widehat J\to J$ \cite[Lemma~11.3.1]{BL04}.
The two isomorphisms are inverse to each other up to sign \cite[Prop.~11.3.5]{BL04}:
\[ \mathrm{A}^\ast\circ\varphi_\Theta=-\mathrm{id}_J. \]
As in \cite{HH21}, we identify $J$ with $\widehat J$ by $-\varphi_\Theta=(\mathrm{A}^\ast)^{-1}$.

The \emph{theta function} of $\Theta$ is a nonzero section $\vartheta\in H^0(J,\mathcal{O}_J(\Theta))$ whose divisor is $\Theta$.
The section $\vartheta$ is unique up to scaling.
Its pullback to the universal cover $\mathbb{C}^2$ of $J$ is, up to a nowhere-vanishing factor, a translate of Riemann's theta function \cite{BL04}.

\paragraph{The Fourier--Mukai transform}
Let $\mathcal{P}$ be the Poincar\'e bundle on $J\times\widehat J\cong J\times J$, normalized to be trivial on $\{\mathcal{O}_{\widetilde{C}_b}\}\times\widehat J$.
For $\mathcal{L}\in J$, the line bundle $P_{\mathcal{L}}=\mathcal{P}|_{\{\mathcal{L}\}\times\widehat J}$, regarded as a line bundle on $J$ via the identification $-\varphi_\Theta$,\footnote{With the identification $\varphi_\Theta=(-\varphi_\Theta)\circ(-1_J)$ instead, one would obtain $\mathrm{A}^\ast P_{\mathcal{L}}\cong\mathcal{L}^{-1}$, since $(-1_J)^\ast M\cong M^{-1}$ for every $M\in\mathrm{Pic}^0(J)$.} is then characterized by the property \cite[(6.4)]{HH21}
\begin{equation}\label{eq-AJP}
\mathrm{A}^\ast P_{\mathcal{L}}\cong\mathcal{L}.
\end{equation}
We denote by $S_J$ the Fourier--Mukai transform on $J$,
\[ S_J(\mathcal{F})=R\mathrm{pr}_{2\ast}(\mathrm{pr}_1^\ast\mathcal{F}\otimes\mathcal{P}). \]
The transform of the skyscraper sheaf at a point $\mathcal{L}$ of $J$ is $P_{\mathcal{L}}$.

Hausel--Hitchin have worked out the Fourier--Mukai transforms of the structure sheaves of the upward flows of very stable Hodge bundles \cite[Thm.~1.5]{HH21}, motivated by mirror symmetry in the classical limit of the geometric Langlands correspondence \cite[Conj.~2.5]{DP12}.
Hausel--Hitchin relate these transforms to the approach of Donagi--Pantev to the correspondence \cite{DP09}.
This approach uses a Fourier--Mukai transform along the Hitchin fibers \cite{DP08, DP12} and has been carried out for rank-two parabolic bundles on $\mathbb{P}^1$ with five marked points \cite{DP22}.

In our setting, let $\mathcal{M}^\circ=h^{-1}(B_s)$, a principally polarized and self-dual abelian scheme over $B_s$, and let $S$ be the relative transform of sheaves on $\mathcal{M}^\circ$ whose kernel is the relative Poincar\'e bundle on $\mathcal{M}^\circ\times_{B_s}\mathcal{M}^\circ$, normalized along the Hitchin section.
Let $\alpha$ be a very stable Hodge bundle.
Since $h$ restricts to a finite and flat morphism $\alpha^+\to B$ (\Cref{prop-verystable}), $S(\mathcal{O}_{\alpha^+\cap\mathcal{M}^\circ})$ is a vector bundle on $\mathcal{M}^\circ$, namely the pushforward of the Poincar\'e bundle along the corresponding finite flat cover.
The restriction of this vector bundle to $J$ is $S_J(\mathcal{O}_{\alpha^+\cap J})$.
\begin{itemize}
\item If $\alpha$ is the uniformizing Higgs bundle, then $\alpha^+=\mathrm{Hit}(B)$ intersects $J$ in $\Gamma_b$, so $S(\mathcal{O}_{\alpha^+\cap\mathcal{M}^\circ})$ is the trivial line bundle.
The restriction of this line bundle to $J$ is $P_{\mathcal{O}_{\widetilde{C}_b}}=\mathcal{O}_J$, which pulls back to the zero momentum sector $\mathcal{O}_{\widetilde{C}_b}$.
In the magnetic field $L_{\mathbf{B}}=K_{\widetilde{C}_b}$, the lowest Landau level of this pullback is the jumped level $H^0(K_{\widetilde{C}_b})\cong\mathbb{C}^2$ of \Cref{tab-fields}.
\item If $\alpha=D=p_1+p_2\in X_1$ is a very stable Hodge bundle, i.e.\ $p_1\neq p_2$, then $S(\mathcal{O}_{D^+\cap\mathcal{M}^\circ})$ has rank $4$.
If $p_1$ and $p_2$ are not branch points of $\pi$, then $D^+\cap J$ consists of the four momenta $\mathcal{L}_1,\dots,\mathcal{L}_4$ of \Cref{lem-BA}(i), with Baker--Akhiezer sections $\widetilde{s}_{i_j}(\phi_j)$ and divisors $\widetilde{D}_1,\dots,\widetilde{D}_4$, so that
\[ S_J(\mathcal{O}_{D^+\cap J})\cong\bigoplus_{j=1}^4P_{\mathcal{L}_{j}}. \]
It follows from \eqref{eq-AJP} that the pullback of this transform to the unit cell is the direct sum of the four momentum sectors.
So in the magnetic field $L_{\mathbf{B}}=K_{\widetilde{C}_b}$, the lowest Landau level of this pullback is the direct sum of those of the four momenta $\mathcal{L}_j$.
This level is spanned by the four states whose vortices lie over $p_1$ and $p_2$:
\[ \mathrm{A}^\ast S_J(\mathcal{O}_{D^+\cap J})\cong\bigoplus_{j=1}^4\mathcal{L}_{j}, \qquad \bigoplus_{j=1}^4H^0(\widetilde{C}_b,\mathcal{L}_{j}\otimes K_{\widetilde{C}_b})=\bigoplus_{j=1}^4\mathbb{C}\widetilde{s}_{i_j}(\phi_j). \]
\end{itemize}

\paragraph{Translates of the theta function}
For $\mathcal{L}\in J$, the translation-invariant line bundle $P_{\mathcal{L}}$ is
\[ P_{\mathcal{L}}\cong t_{-\mathcal{L}}^\ast\mathcal{O}_J(\Theta)\otimes\mathcal{O}_J(-\Theta). \]
Indeed, the right-hand side is $\varphi_\Theta(-\mathcal{L})$, so its pullback along $\mathrm{A}$ is $\mathcal{L}$, as $\mathrm{A}^\ast\circ\varphi_\Theta=-\mathrm{id}_J$.
Since $\mathrm{A}^\ast\colon\widehat J\to J$ is an isomorphism \cite[Lemma~11.3.1]{BL04}, it follows from \eqref{eq-AJP} that $P_{\mathcal{L}}$ is isomorphic to the right-hand side.
Hence $P_{\mathcal{L}}(\Theta)=P_{\mathcal{L}}\otimes\mathcal{O}_J(\Theta)\cong t_{-\mathcal{L}}^\ast\mathcal{O}_J(\Theta)$, whose sections are the scalings of the translate
\[ \vartheta_{\mathcal{L}}=t_{-\mathcal{L}}^\ast\vartheta, \qquad \vartheta_{\mathcal{L}}(z)=\vartheta(z-\mathcal{L}), \qquad \mathrm{div}\,\vartheta_{\mathcal{L}}=\Theta+\mathcal{L}. \]

\paragraph{Restriction to the curve}
Since the canonical bundle of $J$ is trivial, adjunction along $\Theta\cong\widetilde{C}_b$ gives $\mathrm{A}^\ast\mathcal{O}_J(\Theta)\cong K_{\widetilde{C}_b}$.
The pullbacks along $\mathrm{A}$ of the line bundle $P_{\mathcal{L}}(\Theta)$ and of its section $\vartheta_{\mathcal{L}}$ are
\[ \mathrm{A}^\ast\big(P_{\mathcal{L}}(\Theta)\big)\cong\mathcal{L}\otimes K_{\widetilde{C}_b}, \]
\[ \mathrm{A}^\ast\vartheta_{\mathcal{L}}\in H^0(\widetilde{C}_b,\mathcal{L}\otimes K_{\widetilde{C}_b}), \qquad (\mathrm{A}^\ast\vartheta_{\mathcal{L}})(y)=\vartheta(\mathrm{A}(y)-\mathcal{L}), \quad y\in\widetilde{C}_b. \]
Thus $\mathrm{A}^\ast\vartheta_{\mathcal{L}}$ is the restriction of $\vartheta_{\mathcal{L}}$ to $\Theta$, so the zeros of $\mathrm{A}^\ast\vartheta_{\mathcal{L}}$ are the points $y$ with $\mathrm{A}(y)\in\Theta+\mathcal{L}$.

\paragraph{Riemann's vanishing theorem}
If $\mathcal{L}=\mathcal{O}_{\widetilde{C}_b}$, then $\mathrm{A}^\ast\vartheta=0$, as $\vartheta$ vanishes on $\Theta=\mathrm{A}(\widetilde{C}_b)$.
If $\mathcal{L}\neq\mathcal{O}_{\widetilde{C}_b}$, then $\Theta+\mathcal{L}\neq\Theta$, so $\vartheta_{\mathcal{L}}$ does not vanish identically on $\Theta$.
Hence $\mathrm{A}^\ast\vartheta_{\mathcal{L}}$ is a nonzero section of $\mathcal{L}\otimes K_{\widetilde{C}_b}$, whose divisor $\mathrm{A}^\ast(\Theta+\mathcal{L})$ has degree $\Theta\cdot\Theta=2$.
As $h^0(\mathcal{L}\otimes K_{\widetilde{C}_b})=1$ (cf.\ \eqref{eq-dimlinsys}), this divisor is the unique effective divisor in $|\mathcal{L}\otimes K_{\widetilde{C}_b}|$.
It follows from \eqref{eq-LD} that this divisor is the Baker--Akhiezer divisor $\widetilde{D}=\widetilde{p}_1+\widetilde{p}_2$ of the Higgs bundle $(\pi_\ast\mathcal{L},\pi_\ast\lambda)$:
\[ \mathrm{div}\,\mathrm{A}^\ast\vartheta_{\mathcal{L}}=\mathrm{A}^\ast(\Theta+\mathcal{L})=\widetilde{p}_1+\widetilde{p}_2. \]

\paragraph{Theta functions and lowest Landau levels}
For a momentum $\mathcal{L} \in J$, pulling back along the Abel--Jacobi map defines the linear map
\[ \mathrm{A}^\ast\colon H^0(J,P_{\mathcal{L}}(\Theta))\longrightarrow H^0(\widetilde{C}_b,\mathcal{L}\otimes K_{\widetilde{C}_b}) \]
from the multiples of $\vartheta_{\mathcal{L}}$ to the lowest Landau level of momentum $\mathcal{L}$ of the hyperbolic crystal in the magnetic field $L_{\mathbf{B}}=K_{\widetilde{C}_b}$ (\Cref{tab-fields}).

The following facts follow from Riemann's vanishing theorem and \eqref{eq-dimlinsys}.
\begin{itemize}
\item If $\mathcal{L}\neq\mathcal{O}_{\widetilde{C}_b}$, this map is an isomorphism, so the restricted theta function and the Baker--Akhiezer section $\widetilde{s}_i(\phi)$ are proportional,
\begin{equation}\label{eq-theta-BA}
\mathbb{C}\,\mathrm{A}^\ast\vartheta_{\mathcal{L}}=\mathbb{C}\,\widetilde{s}_i(\phi).
\end{equation}
So the vortices $\widetilde{p}_1$ and $\widetilde{p}_2$ of the lowest Landau level states are the zeros of $y\mapsto\vartheta(\mathrm{A}(y)-\mathcal{L})$.

\item At $\Gamma_b$, this map is zero, whereas the lowest Landau level jumps to $H^0(K_{\widetilde{C}_b})\cong\mathbb{C}^2$ (cf.\ \Cref{tab-fields}).
\end{itemize}

Iengo and Li write lowest Landau level states in a magnetic field of $g$ flux quanta as pullbacks of theta functions along the Abel--Jacobi map \cite[\S2]{IL94}.
Their magnetic field is uniform not for the hyperbolic metric but for the Bergman metric, the pullback along $\mathrm{A}$ of the translation-invariant metric on $J$ whose K\"ahler form represents $\theta$.
In their magnetic field, the lowest Landau levels are again the spaces \eqref{eq-LLL} of holomorphic sections.
In genus two, where the flux $g$ equals $2g-2$, their theta functions are, up to nowhere-vanishing factors, the restrictions $\mathrm{A}^\ast\vartheta_{\mathcal{L}}$.
So their theta function vanishes identically exactly at $\Gamma_b$, where the lowest Landau level $H^0(K_{\widetilde{C}_b})$ is two-dimensional.

\subsection{A four-dimensional crystal}\label{sec-4d}

\paragraph{Band theory with the Jacobian as unit cell}
The Brillouin zone $J$ of the hyperbolic crystal is also the Brillouin zone of a four-dimensional Euclidean crystal, whose unit cell is the dual torus $\widehat J=\mathrm{Pic}^0(J)$ \cite{KR22}.
The momentum sectors of this crystal are the flat unitary line bundles on its unit cell: the momentum sector of $\mathcal{L}\in J$ is $P_{\mathcal{L}}$, with its flat unitary connection, so the wave functions of momentum $\mathcal{L}$ are the sections of $P_{\mathcal{L}}$.
As in \Cref{sec-theta}, we identify $\widehat J$ with $J$ via $-\varphi_\Theta$, 
so that the unit cell and the Brillouin zone of this crystal are both $J$, although physically the first is the space of positions and the second is the space of momenta.

\begin{table}[ht]
\centering
\renewcommand{\arraystretch}{1.5}
\begin{tabular}{>{\raggedright\arraybackslash}p{0.45\textwidth}|>{\raggedright\arraybackslash}p{0.45\textwidth}}
\textbf{Hyperbolic crystal} & \textbf{Four-dimensional Euclidean crystal} \\ \hline
unit cell $\widetilde{C}_b$, of genus two, with its hyperbolic metric & unit cell $\widehat J\cong J$, a flat torus of real dimension four \\
Brillouin zone $J=\mathrm{Jac}(\widetilde{C}_b)$ & Brillouin zone $J$ \\
momentum $\mathcal{L}\in J$, with momentum sector $L_\chi$, a flat line bundle on $\widetilde{C}_b$ & momentum $\mathcal{L}\in J$, with momentum sector $P_{\mathcal{L}}$, a flat line bundle on $J$ \\
magnetic field $L_{\mathbf{B}}=K_{\widetilde{C}_b}$, of $2$ flux quanta & magnetic field $L_{\mathbf{B}}=\mathcal{O}_J(\Theta)$, of the principal polarization $\theta$ \\
wave functions in the magnetic field: sections of $\mathcal{L}\otimes K_{\widetilde{C}_b}$ & wave functions in the magnetic field: sections of $P_{\mathcal{L}}(\Theta)$ \\
lowest Landau level $H^0(\widetilde{C}_b,\mathcal{L}\otimes K_{\widetilde{C}_b})$, of dimension $1$, and $2$ at $\Gamma_b$ & lowest Landau level $H^0(J,P_{\mathcal{L}}(\Theta))$, of dimension $1$ at every momentum \\
state: the Baker--Akhiezer section $\widetilde{s}_i(\phi)$, a scaling of $\mathrm{A}^\ast\vartheta_{\mathcal{L}}$ & state: the theta function $\vartheta_{\mathcal{L}}$ \\
vortices: the two points $\widetilde{D}=\mathrm{A}^\ast(\Theta+\mathcal{L})$ ($\mathcal{L}\neq\mathcal{O}_{\widetilde{C}_b}$) & vortices: the curve $\Theta+\mathcal{L}$, a translate of the hyperbolic unit cell $\Theta$, intersecting $\Theta$ in $\mathrm{A}(\widetilde{D})=\Theta\cap(\Theta+\mathcal{L})$ ($\mathcal{L}\neq\mathcal{O}_{\widetilde{C}_b}$)
\end{tabular}
\vspace{8pt} 
\caption{The hyperbolic crystal and the four-dimensional Euclidean crystal of \cite{KR22} with the same Brillouin zone $J$: the first three rows are those of \cite{KR22}, whereas the others describe both crystals in a magnetic field. The two crystals are related by the Abel--Jacobi map $\mathrm{A}\colon\widetilde{C}_b\xrightarrow{\sim}\Theta\subset J$, with $\mathrm{A}^\ast P_{\mathcal{L}}\cong\mathcal{L}$ and $\mathrm{A}^\ast\mathcal{O}_J(\Theta)\cong K_{\widetilde{C}_b}$ as holomorphic line bundles.}\label{tab-crystals}
\end{table}

The hyperbolic unit cell is embedded in the four-dimensional unit cell $\widehat J \equiv J$
as the theta divisor $\Theta=\mathrm{A}(\widetilde{C}_b)$. 
It follows from \eqref{eq-AJP} that 
pulling back along the Abel--Jacobi map $\mathrm{A}$ maps the wave functions of the four-dimensional crystal to those of the hyperbolic crystal, preserving the momentum.

Let us now turn on magnetic fields. 
We put on the four-dimensional unit cell the uniform magnetic field with line bundle
\[ L_{\mathbf{B}}=\mathcal{O}_J(\Theta), \]
endowed with the Hermitian metric whose curvature is $-2\pi i$ times the translation-invariant K\"ahler form on $J$ representing $\theta$.
As for the hyperbolic crystal, recall from \Cref{sec-fields} that 
the wave functions of momentum $\mathcal{L}$ in the magnetic field $L_{\mathbf{B}}=\mathcal{O}_J(\Theta)$ are the sections of $L_{\mathbf{B}}\otimes P_{\mathcal{L}}=P_{\mathcal{L}}(\Theta)$.
It follows from \eqref{eq-LLL} that the lowest Landau level of momentum $\mathcal{L}$ is
\[ H^0(J,P_{\mathcal{L}}(\Theta)), \]
which is spanned by $\vartheta_{\mathcal{L}}$.
In contrast to the hyperbolic crystal setting, the lowest Landau level is now one-dimensional at every momentum, including the zero momentum $\Gamma_b$.

\paragraph{The crystallographic meaning of theta functions}
The four-dimensional crystal suggests an answer to the questions raised in \cite{KR22} about the crystallographic meaning of the principal polarization of the Brillouin zone, of the theta divisor and of theta functions.
From the point of view of the four-dimensional crystal, the principal polarization $\theta$ is the class of the magnetic field $L_{\mathbf{B}}=\mathcal{O}_J(\Theta)$.
The theta divisor is the hyperbolic unit cell embedded by the Abel--Jacobi map.
Theta functions are the lowest Landau level states in this magnetic field.

It follows from \eqref{eq-theta-BA} that the restrictions of these states to the hyperbolic unit cell are the Baker--Akhiezer sections at the nonzero momenta, as holomorphic sections and up to scaling.
At $\Gamma_b$, the restriction is zero, as $\vartheta$ vanishes on $\Theta=\mathrm{A}(\widetilde{C}_b)$.

\FloatBarrier

\section*{Acknowledgements}
\foreignlanguage{vietnamese}{Đ}.Q.D.\ is supported by the German Research Foundation (DFG) via the Walter Benjamin fellowship.
S.R.\ acknowledges partial support from a Natural Sciences and Engineering Research Council of Canada (NSERC) Discovery Grant.

\section*{Tool and computational resource disclosure}
The authors started to look for a specific relation and interpretation of wobbly (Higgs) bundles in the framework of band theory \cite{KR22} in November 2025, 
in particular via the characterization of wobbliness in terms of Baker--Akhiezer divisors following \cite{HH21, Din24}.
We started to work out the nilpotent Hitchin fiber 
and wobbly locus of the moduli of $\mathcal{O}(3)$-twisted Higgs bundles in early summer 2026.
Subsequently, AI assistance (Claude Opus 5 \& 5.5) was used in working out the interpretation of Baker--Akhiezer sections in band theory, in particular the relation to Landau levels, 
in doing computations and in finalizing this paper, 
including making all the figures.
The first-named author carried out these AI experiments; the second-named author used only traditional methods in this work.
Both authors nonetheless take full responsibility for AI usage and can make available specific queries and responses as appropriate.

We have checked all statements and proofs of the paper and take full responsibility for them.

\printbibliography

@article{BB73,
	author  = {Bia{\l}ynicki-Birula, A.},
	title   = {Some theorems on actions of algebraic groups},
	journal = {Ann. of Math. (2)},
	volume  = {98},
	pages   = {480--497},
	year    = {1973}
}

@article{HH21,
	author  = {Hausel, T. and Hitchin, N.},
	title   = {Very stable {H}iggs bundles, equivariant multiplicity and mirror symmetry},
	journal = {Invent. Math.},
	volume  = {228},
	pages   = {893--989},
	year    = {2022},
	doi     = {10.1007/s00222-021-01093-7}
}

@article{HRV15,
	author  = {Hausel, T. and Rodriguez-Villegas, F.},
	title   = {Cohomology of large semiprojective hyperk{\"a}hler varieties},
	journal = {Ast{\'e}risque},
	volume  = {370},
	pages   = {113--156},
	year    = {2015}
}

@article{Hit87a,
	author  = {Hitchin, N.},
	title   = {The self-duality equations on a {R}iemann surface},
	journal = {Proc. London Math. Soc. (3)},
	volume  = {55},
	pages   = {59--126},
	year    = {1987}
}

@article{Laumon88,
	author  = {Laumon, G.},
	title   = {Un analogue global du c{\^o}ne nilpotent},
	journal = {Duke Math. J.},
	volume  = {57},
	number  = {2},
	pages   = {647--671},
	year    = {1988}
}

@article{Markman94,
	author  = {Markman, E.},
	title   = {Spectral curves and integrable systems},
	journal = {Compositio Math.},
	volume  = {93},
	number  = {3},
	pages   = {255--290},
	year    = {1994}
}

@article{Nitsure,
	author  = {Nitsure, N.},
	title   = {Moduli space of semistable pairs on a curve},
	journal = {Proc. London Math. Soc. (3)},
	volume  = {62},
	number  = {2},
	pages   = {275--300},
	year    = {1991}
}

@article{RS18,
	author  = {Rayan, S. and Sundbo, E.},
	title   = {Twisted argyle quivers and {H}iggs bundles},
	journal = {Bull. Sci. Math.},
	volume  = {146},
	pages   = {1--32},
	year    = {2018},
	doi     = {10.1016/j.bulsci.2018.03.003}
}

@article{RS21,
	author  = {Rayan, S. and Sundbo, E.},
	title   = {Twisted cyclic quiver varieties on curves},
	journal = {Eur. J. Math.},
	volume  = {7},
	number  = {1},
	pages   = {205--225},
	year    = {2021}
}

@article{MR21,
	author  = {Maciejko, J. and Rayan, S.},
	title   = {Hyperbolic band theory},
	journal = {Sci. Adv.},
	volume  = {7},
	pages   = {eabe9170},
	year    = {2021},
	doi     = {10.1126/sciadv.abe9170}
}

@article{MR22,
	author  = {Maciejko, J. and Rayan, S.},
	title   = {Automorphic {B}loch theorems for hyperbolic lattices},
	journal = {Proc. Natl. Acad. Sci. USA},
	volume  = {119},
	pages   = {e2116869119},
	year    = {2022},
	doi     = {10.1073/pnas.2116869119}
}

@article{BGKMRT22,
	author  = {Boettcher, I. and Gorshkov, A. V. and Koll{\'a}r, A. J. and Maciejko, J. and Rayan, S. and Thomale, R.},
	title   = {Crystallography of hyperbolic lattices},
	journal = {Phys. Rev. B},
	volume  = {105},
	pages   = {125118},
	year    = {2022},
	doi     = {10.1103/PhysRevB.105.125118}
}

@article{KR22,
	author  = {Kienzle, E. and Rayan, S.},
	title   = {Hyperbolic band theory through {H}iggs bundles},
	journal = {Adv. Math.},
	volume  = {409},
	pages   = {108664},
	year    = {2022},
	doi     = {10.1016/j.aim.2022.108664}
}

@article{NR24,
	author  = {Nagy, {\'A}. and Rayan, S.},
	title   = {On the hyperbolic {B}loch transform},
	journal = {Ann. Henri Poincar{\'e}},
	volume  = {25},
	number  = {3},
	pages   = {1713--1732},
	year    = {2024},
	doi     = {10.1007/s00023-023-01336-8}
}

@unpublished{XMGTRM25,
	author     = {Xu, X. and Mahmoud, A. A. and Gorgichuk, N. and Thomale, R. and Rayan, S. and Mariantoni, M.},
	title      = {A scalable superconducting circuit framework for emulating physics in hyperbolic space},
	eprinttype = {arxiv},
	eprint     = {2510.23827},
	year       = {2025}
}

@unpublished{Din24,
	author     = {{Đinh D.}},
	title      = {Rank-2 wobbly bundles from special divisors on spectral curves},
	eprinttype = {arxiv},
	eprint     = {2406.04224},
	year       = {2024}
}

@article{BNR89,
	author  = {Beauville, A. and Narasimhan, M. S. and Ramanan, S.},
	title   = {Spectral curves and the generalised theta divisor},
	journal = {J. Reine Angew. Math.},
	volume  = {398},
	pages   = {169--179},
	year    = {1989}
}

@article{Kri77,
	author  = {Krichever, I. M.},
	title   = {Methods of algebraic geometry in the theory of non-linear equations},
	journal = {Russian Math. Surveys},
	volume  = {32},
	number  = {6},
	pages   = {185--213},
	year    = {1977}
}

@article{Dub81,
	author  = {Dubrovin, B. A.},
	title   = {Theta functions and non-linear equations},
	journal = {Russian Math. Surveys},
	volume  = {36},
	number  = {2},
	pages   = {11--92},
	year    = {1981},
	doi     = {10.1070/RM1981v036n02ABEH002596}
}

@article{TKNN82,
	author  = {Thouless, D. J. and Kohmoto, M. and Nightingale, M. P. and den Nijs, M.},
	title   = {Quantized {H}all conductance in a two-dimensional periodic potential},
	journal = {Phys. Rev. Lett.},
	volume  = {49},
	pages   = {405--408},
	year    = {1982}
}

@article{IL94,
	author  = {Iengo, R. and Li, D.},
	title   = {Quantum mechanics and quantum {H}all effect on {R}iemann surfaces},
	journal = {Nuclear Phys. B},
	volume  = {413},
	pages   = {735--753},
	year    = {1994},
	eprinttype = {arxiv},
	eprint  = {hep-th/9307011}
}

@incollection{Kle16,
	author    = {Klevtsov, S.},
	title     = {Geometry and large {N} limits in {L}aughlin states},
	booktitle = {Travaux math\'ematiques},
	volume    = {24},
	pages     = {63--127},
	publisher = {Univ. Luxembourg},
	year      = {2016},
	eprinttype = {arxiv},
	eprint    = {1608.02928}
}

@article{ASZ94,
	author  = {Avron, J. E. and Seiler, R. and Zograf, P. G.},
	title   = {Adiabatic quantum transport: quantization and fluctuations},
	journal = {Phys. Rev. Lett.},
	volume  = {73},
	pages   = {3255--3257},
	year    = {1994},
	eprinttype = {arxiv},
	eprint  = {hep-lat/9405017}
}

@misc{Ray26,
	author        = {Rayan, S.},
	title         = {Beyond {$K$}-theory: geometry and holomorphy in hyperbolic band theory},
	year          = {2026},
	eprint        = {2608.01596},
	archivePrefix = {arXiv},
	note          = {Preprint}
}

@book{BL04,
	author    = {Birkenhake, C. and Lange, H.},
	title     = {Complex Abelian Varieties},
	edition   = {2},
	series    = {Grundlehren Math. Wiss.},
	volume    = {302},
	publisher = {Springer},
	address   = {Berlin},
	year      = {2004}
}

@book{AM76,
	author    = {Ashcroft, N. W. and Mermin, N. D.},
	title     = {Solid State Physics},
	publisher = {Holt, Rinehart and Winston},
	address   = {New York},
	year      = {1976}
}

@article{Kuc16,
	author  = {Kuchment, P.},
	title   = {An overview of periodic elliptic operators},
	journal = {Bull. Amer. Math. Soc.},
	volume  = {53},
	number  = {3},
	pages   = {343--414},
	year    = {2016}
}

@article{HR85,
	author  = {Haldane, F. D. M. and Rezayi, E. H.},
	title   = {Periodic {L}aughlin-{J}astrow wave functions for the fractional quantized {H}all effect},
	journal = {Phys. Rev. B},
	volume  = {31},
	pages   = {2529--2531},
	year    = {1985},
	doi     = {10.1103/PhysRevB.31.2529}
}

@article{KFH19,
	author  = {Koll\'ar, A. J. and Fitzpatrick, M. and Houck, A. A.},
	title   = {Hyperbolic lattices in circuit quantum electrodynamics},
	journal = {Nature},
	volume  = {571},
	number  = {7763},
	pages   = {45--50},
	year    = {2019},
	doi     = {10.1038/s41586-019-1348-3}
}

@article{Com87,
	author  = {Comtet, A.},
	title   = {On the {L}andau levels on the hyperbolic plane},
	journal = {Ann. Physics},
	volume  = {173},
	pages   = {185},
	year    = {1987},
	doi     = {10.1016/0003-4916(87)90098-4}
}

@unpublished{RS26,
	author        = {Rayan, S. and Sundbo, E.},
	title         = {Quantum error correcting codes from quiver moduli: {G}rassmannians, cyclic monodromy, and argyle collision divisors},
	year          = {2026},
	note          = {In preparation}
}

@unpublished{DDNP26,
	author        = {{Đinh D.} and Donagi, R. and Nollau, C. and Pantev, T.},
	title         = {Shaky locus in semiprojective variety from {B}ia{\l}ynicki-{B}irula stratification},
	year          = {2026},
	note          = {In preparation}
}

@unpublished{DT23,
	author     = {{Đinh D.} and Teschner, J.},
	title      = {Classical limit of the geometric {L}anglands correspondence for {$SL(2,\mathbb{C})$}},
	eprinttype = {arxiv},
	eprint     = {2312.13393},
	year       = {2023}
}

@unpublished{DFT26,
	author     = {{Đinh D.} and Figiel, T. and Teschner, J.},
	title      = {Stratification of opers with apparent singularities and isomonodromic bubbling},
	eprinttype = {arxiv},
	eprint     = {2610.22326},
	year       = {2026}
}

@article{Bottacin95,
	author  = {Bottacin, F.},
	title   = {Symplectic geometry on moduli spaces of stable pairs},
	journal = {Ann. Sci. \'Ecole Norm. Sup. (4)},
	volume  = {28},
	number  = {4},
	pages   = {391--433},
	year    = {1995}
}

@incollection{DP09,
	author    = {Donagi, R. and Pantev, T.},
	title     = {Geometric {L}anglands and non-abelian {H}odge theory},
	booktitle = {Surveys in Differential Geometry {XIII}},
	pages     = {85--116},
	publisher = {International Press},
	address   = {Somerville, MA},
	year      = {2009}
}

@article{PPe19,
	author  = {Pauly, C. and Pe{\'o}n-Nieto, A.},
	title   = {Very stable bundles and properness of the {H}itchin map},
	journal = {Geom. Dedicata},
	volume  = {198},
	pages   = {143--148},
	year    = {2019}
}

@misc{MSTR26,
	author        = {Mahmoud, A. A. and Sorbello, R. and Thomale, R. and Rayan, S.},
	title         = {Quantum geometry in hyperbolic band theory},
	year          = {2026},
	eprint        = {2609.14876},
	archivePrefix = {arXiv},
	note          = {Preprint}
}

@incollection{Ray26b,
	author    = {Rayan, S.},
	title     = {Quantum matter from algebraic geometry and number theory},
	booktitle = {Nankai Symposium on Mathematical Dialogues},
	editor    = {He, Y.-H. and Ge, M.-L. and Bai, C.-M. and Bao, J. and Hirst, E.},
	publisher = {Springer},
	address   = {Singapore},
	pages     = {289--292},
	year      = {2026},
	doi       = {10.1007/978-981-19-2328-9_33}
}

@article{DK10,
	author  = {Douglas, M. R. and Klevtsov, S.},
	title   = {Bergman kernel from path integral},
	journal = {Comm. Math. Phys.},
	volume  = {293},
	number  = {1},
	pages   = {205--230},
	year    = {2010},
	doi     = {10.1007/s00220-009-0915-0}
}

@misc{CL25,
	author        = {Charles, L. and Lefeuvre, T.},
	title         = {Semiclassical defect measures of magnetic {L}aplacians on hyperbolic surfaces},
	year          = {2025},
	eprint        = {2505.08584},
	archivePrefix = {arXiv},
	note          = {Preprint}
}

@book{DP08,
	author    = {Donagi, R. and Pantev, T.},
	title     = {Torus fibrations, gerbes, and duality},
	series    = {Mem. Amer. Math. Soc.},
	volume    = {193},
	number    = {901},
	publisher = {American Mathematical Society},
	address   = {Providence, RI},
	year      = {2008},
	note      = {With an appendix by D.~Arinkin}
}

@article{DP12,
	author  = {Donagi, R. and Pantev, T.},
	title   = {Langlands duality for {H}itchin systems},
	journal = {Invent. Math.},
	volume  = {189},
	number  = {3},
	pages   = {653--735},
	year    = {2012}
}

@book{DP22,
	author    = {Donagi, R. and Pantev, T.},
	title     = {Parabolic {H}ecke eigensheaves},
	series    = {Ast\'erisque},
	volume    = {435},
	publisher = {Soci\'et\'e Math\'ematique de France},
	address   = {Paris},
	year      = {2022}
}

@article{Len22,
	author  = {Lenggenhager, P. M. and Stegmaier, A. and Upreti, L. K. and Hofmann, T. and Helbig, T. and Vollhardt, A. and Greiter, M. and Lee, C. H. and Imhof, S. and Brand, H. and Kie{\ss}ling, T. and Boettcher, I. and Neupert, T. and Thomale, R. and Bzdu{\v{s}}ek, T.},
	title   = {Simulating hyperbolic space on a circuit board},
	journal = {Nat. Commun.},
	volume  = {13},
	pages   = {4373},
	year    = {2022},
	doi     = {10.1038/s41467-022-32042-4}
}

@article{ZYSSZ22,
	author  = {Zhang, W. and Yuan, H. and Sun, N. and Sun, H. and Zhang, X.},
	title   = {Observation of novel topological states in hyperbolic lattices},
	journal = {Nat. Commun.},
	volume  = {13},
	pages   = {2937},
	year    = {2022},
	doi     = {10.1038/s41467-022-30631-x}
}

@article{Che23,
	author  = {Chen, A. and Brand, H. and Helbig, T. and Hofmann, T. and Imhof, S. and Fritzsche, A. and Kie{\ss}ling, T. and Stegmaier, A. and Upreti, L. K. and Neupert, T. and Bzdu{\v{s}}ek, T. and Greiter, M. and Thomale, R. and Boettcher, I.},
	title   = {Hyperbolic matter in electrical circuits with tunable complex phases},
	journal = {Nat. Commun.},
	volume  = {14},
	pages   = {622},
	year    = {2023},
	doi     = {10.1038/s41467-023-36359-6}
}

@article{ZDZSZ23,
	author  = {Zhang, W. and Di, F. and Zheng, X. and Sun, H. and Zhang, X.},
	title   = {Hyperbolic band topology with non-trivial second {C}hern numbers},
	journal = {Nat. Commun.},
	volume  = {14},
	pages   = {1083},
	year    = {2023},
	doi     = {10.1038/s41467-023-36767-8}
}

@article{Hua24,
	author  = {Huang, L. and He, L. and Zhang, W. and Zhang, H. and Liu, D. and Feng, X. and Liu, F. and Cui, K. and Huang, Y. and Zhang, W. and Zhang, X.},
	title   = {Hyperbolic photonic topological insulators},
	journal = {Nat. Commun.},
	volume  = {15},
	year    = {2024},
	doi     = {10.1038/s41467-024-46035-y}
}

@article{IAM21,
	author  = {Ikeda, K. and Aoki, S. and Matsuki, Y.},
	title   = {Hyperbolic band theory under magnetic field and {D}irac cones on a higher genus surface},
	journal = {J. Phys.: Condens. Matter},
	volume  = {33},
	pages   = {485602},
	year    = {2021},
	doi     = {10.1088/1361-648X/ac24c4}
}

@article{SUTB22,
	author  = {Stegmaier, A. and Upreti, L. K. and Thomale, R. and Boettcher, I.},
	title   = {Universality of {H}ofstadter butterflies on hyperbolic lattices},
	journal = {Phys. Rev. Lett.},
	volume  = {128},
	pages   = {166402},
	year    = {2022},
	doi     = {10.1103/PhysRevLett.128.166402}
}

@article{Res92,
	author  = {Reshetikhin, N.},
	title   = {The {K}nizhnik--{Z}amolodchikov system as a deformation of the isomonodromy problem},
	journal = {Lett. Math. Phys.},
	volume  = {26},
	pages   = {167--177},
	year    = {1992},
	doi     = {10.1007/BF00404594}
}

@unpublished{Tes17,
	author     = {Teschner, J.},
	title      = {Classical conformal blocks and isomonodromic deformations},
	eprinttype = {arxiv},
	eprint     = {1707.07968},
	year       = {2017}
}

@article{KMMW17,
	author  = {Klevtsov, S. and Ma, X. and Marinescu, G. and Wiegmann, P.},
	title   = {Quantum {H}all effect and {Q}uillen metric},
	journal = {Comm. Math. Phys.},
	volume  = {349},
	number  = {3},
	pages   = {819--856},
	year    = {2017},
	doi     = {10.1007/s00220-016-2789-2}
}

@book{ACGH85,
	author    = {Arbarello, E. and Cornalba, M. and Griffiths, P. A. and Harris, J.},
	title     = {Geometry of Algebraic Curves. {V}ol. {I}},
	series    = {Grundlehren Math. Wiss.},
	volume    = {267},
	publisher = {Springer-Verlag},
	address   = {New York},
	year      = {1985}
}

@online{GLC,
	author  = {Gaitsgory, D. and Raskin, S. and others},
	title   = {Proof of the geometric {L}anglands conjecture},
	url     = {https://people.mpim-bonn.mpg.de/gaitsgde/GLC/},
	year    = {2024}
}

@unpublished{Ari16,
	author     = {Arinkin, D.},
	title      = {Irreducible connections admit generic oper structures},
	eprinttype = {arxiv},
	eprint     = {1602.08989},
	year       = {2016}
}

\end{document}